\documentclass[11pt]{article}
\usepackage[utf8]{inputenc}
\usepackage{palatino}
\usepackage{fullpage}
\usepackage[backref=page]{hyperref}
\hypersetup{
    unicode=false,          
    colorlinks=true,        
    linkcolor=blue,         
    citecolor=purple,       
    filecolor=magenta,      
    urlcolor=cyan           
}

\usepackage{algorithm}
\usepackage[noend]{algorithmic}
\usepackage{booktabs}       
\usepackage{nicefrac}       
\usepackage{microtype}      
\usepackage{xcolor}         
\usepackage{color}
\usepackage{mathtools}
\usepackage{amsfonts}
\usepackage{amsthm}
\usepackage{amssymb}
\usepackage{url}
\usepackage[capitalize,noabbrev,nameinlink]{cleveref}
\usepackage{multicol, multirow}
\usepackage{xfrac}
\renewcommand{\algorithmiccomment}[1]{\bgroup\hfill$\rhd$~\footnotesize{\textcolor{violet}{#1}}\egroup}
\usepackage[group-separator={,}, group-minimum-digits=4]{siunitx}
\usepackage[font=small]{caption}
\usepackage{enumitem}
\usepackage{subcaption}

\theoremstyle{plain}
\newtheorem{theorem}{Theorem}[section]
\newtheorem{proposition}[theorem]{Proposition}
\newtheorem{lemma}[theorem]{Lemma}
\newtheorem{corollary}[theorem]{Corollary}
\theoremstyle{definition}
\newtheorem{definition}[theorem]{Definition}

\theoremstyle{remark}

\newcommand{\eps}{\varepsilon}
\newcommand{\ouralgo}{\texttt{MAD}}
\newcommand{\ouralgolong}{\texttt{MaxAdaptiveDegree}}
\newcommand{\ouralgotworounds}{\texttt{MAD2R}}
\newcommand{\ouralgotworoundslong}{\texttt{MaxAdaptiveDegreeTwoRounds}}
\newcommand{\userweights}{\texttt{UserWeights}}
\newcommand{\selectalgo}{\texttt{WeightAndThreshold}}
\newcommand{\weightalgo}{\texttt{ALG}}
\newcommand{\basicalgo}{\texttt{Basic}}
\newcommand{\sets}{\mathcal{S}}

\newcommand{\dmax}{d_{max}}
\newcommand{\bmin}{b_{min}}
\newcommand{\bmax}{b_{max}}

\newcommand{\U}{\mathcal{U}}
\newcommand{\pluseq}{\mathrel{+}=}
\newcommand\err[1]{{\scriptstyle (\pm #1)}}

\newcommand{\RETURN}{\STATE \textbf{return} }

\newcommand{\DeclareAutoPairedDelimiter}[3]{%
  \expandafter\DeclarePairedDelimiter\csname Auto\string#1\endcsname{#2}{#3}%
  \begingroup\edef\x{\endgroup
    \noexpand\DeclareRobustCommand{\noexpand#1}{%
      \expandafter\noexpand\csname Auto\string#1\endcsname*}}%
  \x}
\DeclareAutoPairedDelimiter\p{\lparen}{\rparen}
\DeclareAutoPairedDelimiter\br{\lbrack}{\rbrack}
\DeclareAutoPairedDelimiter\bc{\lbrace}{\rbrace}
\DeclareAutoPairedDelimiter\angle{\langle}{\rangle}
\DeclareAutoPairedDelimiter\ceil{\lceil}{\rceil}
\DeclareAutoPairedDelimiter\floor{\lfloor}{\rfloor}
\DeclareAutoPairedDelimiter\abs{\lvert}{\rvert}   
\DeclareAutoPairedDelimiter\norm{\lVert}{\rVert}  

\newif\iftodos
\todosfalse

\iftodos
\newcommand\todo[1]{\textcolor{red}{TODO: #1}}
\newcommand\morteza[1]{\textcolor{olive}{Morteza: #1}}
\newcommand\mortezaEdit[1]{\textcolor{olive}{#1}}
\newcommand\justin[1]{\textcolor{blue}{Justin: #1}}
\newcommand\ale[1]{\textcolor{cyan}{Ale: #1}}
\newcommand\vincent[1]{\textcolor{pink}{V: #1}}
\else 
\newcommand\todo[1]{}
\newcommand\morteza[1]{}
\newcommand\mortezaEdit[1]{}
\newcommand\justin[1]{}
\newcommand\ale[1]{}
\newcommand\vincent[1]{}
\fi

\title{Scalable Private Partition Selection via Adaptive Weighting}

\author{%
  Justin Y.\ Chen\thanks{Work done as a student researcher at Google Research.} \\
  MIT \\
  \texttt{justc@mit.edu} \\
  \and
  Vincent Cohen-Addad \\
  Google Research\\
  \texttt{cohenaddad@google.com} \\
  \and
  Alessandro Epasto\\
  Google Research \\
  \texttt{aepasto@google.com} \\
  \and
  Morteza Zadimoghaddam \\
  Google Research \\
  \texttt{zadim@google.com} \\
}

\begin{document}

\maketitle

\begin{abstract}
In the differentially private partition selection problem (a.k.a. private set union, private key discovery), users hold subsets of items from an unbounded universe. The goal is to output as many items as possible from the union of the users' sets while maintaining user-level differential privacy. Solutions to this problem are a core building block for many privacy-preserving ML applications including vocabulary extraction in a private corpus, computing statistics over categorical data and learning embeddings over user-provided items.  
We propose an algorithm for this problem, \texttt{MaxAdaptiveDegree} (\texttt{MAD}), which adaptively reroutes weight from items with weight far above the threshold needed for privacy to items with smaller weight, thereby increasing the probability that less frequent items are output.  Our algorithm can be efficiently implemented in massively parallel computation systems allowing scalability to very large datasets. We prove that our algorithm stochastically dominates the standard parallel algorithm for this problem. We also develop a two-round version of our algorithm, \texttt{MAD2R}, where results of the computation in the first round are used to bias the weighting in the second round to maximize the number of items output. In experiments, our algorithms provide the best results among parallel algorithms and scale to datasets with hundreds of billions of items, up to three orders of magnitude larger than those analyzed by prior sequential algorithms.
\end{abstract}

\setcounter{page}{0}
\thispagestyle{empty}

\newpage
\tableofcontents
\setcounter{page}{0}
\thispagestyle{empty}

\newpage

\section{Introduction}\label{sec:intro}

The availability of large amounts of user data has been one of the driving factors for the widespread adoption and rapid development of modern machine learning and data analytics. 
Consider the example of a system releasing information on the queries asked to a search engine over a period of time~\cite{korolova2009releasing, bavadekar2021google}. Such a system can provide valuable insights to researchers and the public (for instance on health concerns~\cite{bavadekar2021google}) but care is needed in ensuring that the queries output do not leak private and sensitive user information.

In this paper, we focus on the problem of {\it private partition selection}~\cite{desfontaines2022dppartition,gopi2020dpunion} which models the challenge of  extracting as much data as possible from such a dataset while respecting user privacy. 
More formally, the setting of the problem  (which is also known as private set union or private key discovery) is that each user has a private subset of items (e.g., the queries issued by the user) from an unknown and unbounded universe of items (e.g., all strings). The goal is to output as many of the items in the users' sets as possible (i.e., the queries issues by the users), while  providing a strong notion of privacy---User-level Differential Privacy (DP)~\cite{dwork2014book}.

Private partition selection models many challenges beyond the example above, including the problem  of extracting the vocabulary (words, tokens or $n$-gram) present in a private corpus~\cite{zhang2022federated,kim2021differentially}. This task is a fundamental prerequisite for many privacy-preserving natural language processing algorithms~\cite{wilson2019differentially,gopi2020dpunion}, including for training language models for sentence completion and response generation for emails~\cite{kim2021differentially}. Similarly,  learning embedding models over categorical data often requires to identify the categories present in a private dataset~\cite{ghazi2023sparsitypreserving, ghazi2024dpsparsegrad}.
Partition selection underpins many other applications including analyzing private streams~\cite{cardoso2022differentially,zhang2023differentially}, learning sparse histograms~\cite{boneh2021lightweight}, answering SQL queries~\cite{desfontaines2022dppartition} and sparsifying the gradients in the DP SGD method~\cite{ghazi2023sparsitypreserving}.
Unsurprisingly given these applications, private partition selection algorithms \cite{korolova2009releasing} are a core building block of many standard differentially private libraries e.g., PyDP~\cite{pydp_partition}, Google's DP Libraries~\cite{googledplibrary,amin2022plume}, and OpenMined DP Library~\cite{openmined}. 

Real-world datasets for these applications can be massive, potentially containing hundreds of billions of data points, thus requiring algorithms for partition selection that can be efficiently run in large-scale data processing infrastructures---e.g., MapReduce~\cite{dean2004mapreduce}, Hadoop~\cite{hadoop},  Spark~\cite{zaharia2016spark}. In our work, we design a highly parallelizable algorithm for this problem which requires constant parallel rounds in the Massively Parallel Computing model~\cite{karloff2010model} and does not assume to fit the input in memory.  This contrasts to prior algorithms such as \cite{gopi2020dpunion, carvalho2022incorporatingitem} which all require (with the key exception of the uniform weighting method described below) to process all the data sequentially on a single machine and assume storing the input in-memory thus precluding efficient parallelization.

\subsection{Weight and Threshold Approach}
Before introducing our algorithm, we review the popular weighting-based approach to partition selection which is used in many algorithms~\cite{korolova2009releasing, gopi2020dpunion, carvalho2022incorporatingitem, swanberg2023dpsips}. This approach is of interest in the context of large-scale data as some of its variants can be parallelized efficiently~\cite{korolova2009releasing,swanberg2023dpsips}. 

Notice that differential privacy imposes to not output any item which is owned by only a single user. However, it is possible for a private algorithm to output items which appear in {\it many} different sets. This intuition is at the basis of the weighting-based algorithms.

Algorithms in this framework start by subsampling each user's set to bound the maximum number of items per user. Then, these algorithms proceed by increasing, for each user, the weight associated the items present in the user data. Finally, the algorithm adds Gaussian (or Laplace) noise to the total accumulated weight of each item, and outputting all items with noised weight above a certain threshold~\cite{korolova2009releasing, gopi2020dpunion, carvalho2022incorporatingitem, swanberg2023dpsips}. The amount of noise and value of the threshold depends on the privacy parameters and  \emph{crucially} on the sensitivity of the weighting function to the addition or removal of any individual user's set. Loosely speaking, the contribution of each user to the item weights must be bounded in order to achieve differential privacy. Algorithms within this framework differ in the choice of how to assign item weights, but in all designs the key goal is that of limiting the sensitivity of the weighting function.

A basic strategy is {\it uniform weighting}~\cite{korolova2009releasing} where each user  contributes equal weight to each of the items in their set. It is easy to bound the sensitivity of this basic weighting and thus to prove differential privacy. Because of its simplicity, the basic uniform weighting algorithm is extremely parallelizable requiring only basic counting operations over the items in the data.  

Unfortunately, however, uniform weighting is lossy in that it may overallocate weight far above the threshold to high frequency items, missing an opportunity to boost the weight of items closer to the decision boundary. This has inspired the design of greedy weighting schemes such as~\cite{gopi2020dpunion, carvalho2022incorporatingitem} where each user's allocations depend on data of previously analyzed users. All of these algorithms are inherently sequential and require memory proportional to the items present in the data. 

To our knowledge, the uniform weighting~\cite{korolova2009releasing} is essentially the only known solution to the private partition selection problem which is amenable to implementation in a massively parallel computation framework. The sole exception is the scalable, iterative partition selection (DP-SIPS) scheme of~\cite{swanberg2023dpsips} which has as core computation  repeated invocations of the uniform weighting algorithm.

\subsection{Our Contributions}

In this work, we design the first, adaptive, non-uniform weighting algorithm that is amenable to massively parellel implementations. Our algorithm, called \ouralgolong{} (\ouralgo{}), requires linear work in the size of the input and can be implemented in a constant number of rounds in a parallel framework. 
From a technical point of view, the algorithm is based on a careful rerouting of overallocated weight to less frequent items, that together with a delicate sensitivity analysis shows no privacy loss compared to uniform weighting. This means that---given the same privacy parameters---both algorithms utilize exactly the same amount of noise and the same threshold (but our algorithm can better allocate the weight). As a result, we are able to prove that our algorithm stochastically dominates the basic, uniform weighting strategy.

We extend our result to multiple rounds in \ouralgotworoundslong{} (\ouralgotworounds{}), splitting our privacy budget across the rounds, running \ouralgo{} in each round, and outputting the union of items found in both rounds.
Similar to DP-SIPS, in the second round, we remove from the input any items found in the first round (this is private by post-processing).
By a careful generalization of the privacy analysis of the weight and threshold approach, we show that it is possible to also use the noisy weights from the prior round.
We leverage this in two ways. First, we additionally remove items which have very small weights from the first round--these have little chance of being output in the second round. Second, we bias the weighting produced by \ouralgo{} in the second round to further limit overallocation to items which received large weights in the first round. The combination of these ideas yields significant empirical improvements over both the basic algorithm and DP-SIPS.

In \ouralgo{}, users with a too small or large of a cardinality (we equivalently refer to this as the user's degree) are labeled non-adaptive: these users will add uniform weight to their items (or biased weight in the case of \ouralgotworounds{}). 
The rest of the users participate in adaptive reweighting with the privacy analysis making use of the upper bound on their degree. Initially each of these users sends a small amount of weight uniformly among their items (the total amount of weight sent per user is bounded by $1$ rather than the square root of their degree, which is the case for the basic weighting algorithm). Then, items with weight significantly above the threshold are truncated to only have weight slightly above the threshold (we do not want to truncate all the way to the threshold as the added noise can decrease weights). The weight removed via truncation is returned to each user proportional to their initial contributions. Then, users reroute a carefully chosen fraction of this ``excess'' weight back to their items. Finally, users add additional uniform weight to their items to make up for the small amount of weight that was initially sent.

Bounding the sensitivity of \ouralgo{} requires a careful analysis and is significantly more involved than for basic weighting. Several design choices made in our algorithm, such as using an initial uniform weighting inversely proportional to cardinality rather than square root of cardinality, using a minimum and maximum adaptive degree, and choosing the fraction of how much excess weight to reroute are all required for the following theorem to hold. 
Furthermore, we generalize the analysis of the weight and threshold paradigm to allow us to use noisy weights from first round in biasing weights in the second round of \ouralgotworounds{}.
This biasing further complicates the sensitivity analysis of \ouralgo{} which we address by putting limits on the minimum and maximum bias.

\begin{theorem}[Privacy, Informal version of \cref{thm:dp-two-rounds} and \cref{thm:dp-one-round}]
Using \ouralgo{} as the weighting algorithm achieves $(\eps, \delta)$-DP with the exact same noise and threshold parameters as the basic algorithm.
Running \ouralgo{} in two rounds with biases via \ouralgotworounds{} is $(\eps, \delta)$-DP.
\end{theorem}

Within \ouralgo{}, items have their weight truncated if it exceeds an ``adaptive threshold'' $\tau$ after adding the initial weights. $\tau$ is set to be $\beta$ standard deviations of the noise above the true threshold that will be used to determine the output where $\beta \geq 0$ is a free parameter of the algorithm. By design, before adding noise, every item which receives at least weight $\tau$ in the basic algorithm will also receive weight at least $\tau$ by \ouralgo{}. Furthermore, the weights on all other items will only be increased under \ouralgo{} compared to the basic algorithm. Taking the final step of adding noise, we show the following theorem.

\begin{theorem}[Stochastic Dominance, Informal version of \cref{thm:maxalgo-perf}]\label{thm:maxalgo-perf-informal}
Let $U$ be the set of items output when using the basic algorithm and let $U^*$ be the set of items output when using \ouralgo{} as the weighting algorithm.
Then, for items $i \in \U$ and a free parameter $\beta \geq 0$,
\begin{itemize}
    \item If $\Pr\p{i \in U} < \Phi(\beta)$, then $\Pr\p{i \in U^*} \geq \Pr\p{i \in U}$.
    \item Otherwise, $\Pr\p{i \in U^*} \geq \Phi(\beta)$.
\end{itemize}
where $\Phi$ is the standard Gaussian cdf.
\end{theorem}

Compared to the basic algorithm, \ouralgo{} has a higher probability of outputting any item that does not reach the adaptive threshold in its initial stage as it reroutes excess weight to these items.
The theorem shows that \ouralgo{} \emph{stochastically dominates} the basic algorithm on these items.
For the remaining items, they already have an overwhelming probability of being output as their final weight before adding noise is at least several standard deviations above the threshold (this is quantitatively controlled by the parameter $\beta$).
In \cref{sec:utility}, we also describe a simple, concrete family of instances where \ouralgo{} significantly improves upon the baselines.

Finally, we conduct experiments on several publicly-available datasets with up to 800 billions of (user, item) pairs (up to three orders of magnitude larger than prior datasets used in sequential algorithms). Our algorithm outperforms scalable baselines and is competitive with the sequential baselines.

\subsection{Related Work}
Our algorithms are in the area of privacy preserving algorithms with differential privacy guarantee which is the de facto standard of privacy (we refer to~\cite{dwork2014book} for an introduction to this area). As we covered the application and prior work on private partition selection in the introduction, we now provide more details on the work most related to our paper.

The differentially private partition selection problem was first studied in~\cite{korolova2009releasing}. They utilized the now-standard approach of subsampling to limit the number of items in each user's set, constructing weights over items, and thresholding noised weights to produce the output. They proposed a version of the basic weighting algorithm which uses the Laplace mechanism rather than the Gaussian mechanism. This algorithm was also used in~\cite{wilson2020dpsql} within the context of a private SQL system.
The problem received renewed study in~\cite{gopi2020dpunion} where the authors propose a generic class of greedy, sequential weighting algorithms which empirically outperform basic weighting (with either the Laplace or Gaussian mechanism). \cite{carvalho2022incorporatingitem} gave an alternative greedy, sequential weighting algorithm which leverages item frequencies in cases where each user has a multiset of items.
\cite{desfontaines2022dppartition} analyzed in depth the optimal strategy when each user has only a single item (all sets have cardinality one). This is the only work that does not utilize the weight and threshold approach, but it is tailored only for this special case.
The work most related to ours is DP-SIPS~\cite{swanberg2023dpsips} which proposes the first algorithm other than basic weighting which is amenable to implementation in a parallel environment. DP-SIPS splits the privacy budget over a small number of rounds, runs the basic algorithm as a black box each round, and iteratively removes the items found in previous rounds for future computations. This simple idea leads to large empirical improvements, giving a scalable algorithm that has competitive performance with sequential algorithms.



\section{Preliminaries}\label{sec:prelim}
\begin{definition}[Differentially-Private Partition Selection]
In the differentially-private partition selection (a.k.a. private set union or key selection) problem, there are $n$ users with each user $u$ having a set $S_u$ of items from an unknown and possibly infinite universe $\Sigma$ of items: the input is of the form $\sets = \{(u, S_u)\}_{u \in[n]}$. The goal is to output a set of items $U$ of maximum cardinality, such that $U$ is a subset of the union of the users' sets $\U=\cup_{u \in [n]} S_u$,  while maintaining user-level differentially privacy.
\end{definition}
As standard in prior work~\cite{korolova2009releasing, gopi2020dpunion, carvalho2022incorporatingitem, swanberg2023dpsips} we consider the central differential privacy model, where the input data is available to a curator that runs the algorithm and wants to ensure differential privacy for the output of the algorithm. We now formally define these notions.

\begin{definition}[Neighboring Datasets]
We say that two input datasets $\sets$ and $\sets'$ are neighboring if one can be obtained by removing a single user's set from the other, i.e., $\sets' = \sets \cup \{(v, S_v)\}$ for some new user $v$.
\end{definition}

\begin{definition}[Differential Privacy \cite{dwork2014book}]
A randomized algorithm $\mathcal{M}$ is $(\eps, \delta)$-differentially private, or $(\eps, \delta)$-DP, if for any two neighboring datasets $\sets$ and $\sets'$ and for any possible subset of outputs $\mathcal{O} \subseteq \{U: U \subseteq \Sigma \}$,
\[
    \Pr(\mathcal{M}(\sets) \in \mathcal{O}) \leq e^\eps \cdot \Pr(\mathcal{M}(\sets') \in \mathcal{O}) + \delta.
\]
\end{definition}

Let $\Phi: \mathbb{R} \to \mathbb{R}$ be the standard Gaussian cumulative density function.

\begin{proposition}[Gaussian Mechanism \cite{balle2018gaussianmechanism}]\label{prop:gaussian}
Let $f: \mathcal{D} \to \mathbb{R}^d$ be a function with $\ell_2$ sensitivity $\Delta_2$. For any $\eps > 0$ and $\delta \in (0,1]$, the mechanism $M(x) = f(x) + Z$ with $Z \sim \mathcal{N}(0, \sigma^2 I)$ is $(\eps, \delta)$-DP if
\begin{equation*}
    \Phi\left(\frac{\Delta_2}{2\sigma} - \frac{\eps \sigma}{\Delta_2}\right) - e^{\eps} \Phi\left(-\frac{\Delta_2}{2\sigma} - \frac{\eps \sigma}{\Delta_2}\right) \leq \delta.
\end{equation*}
\end{proposition}

\section{Weight and Threshold Meta-Algorithm}\label{sec:meta-algo}
\begin{algorithm}[ht]
\footnotesize
\caption{Meta-algorithm for private partition selection.\\
\selectalgo$(\sets, \eps, \delta, \Delta_0, \weightalgo, h)$
 \newline  {\bf Input: } User sets $\sets=\{(u, S_u)\}_{u\in [n]}$, privacy parameters ($\eps$, $\delta$), degree cap $\Delta_0$, weighting algorithm \weightalgo, upper bound on the novel $\ell_\infty$ sensitivity, function $h: \mathbb{N} \rightarrow \mathbb{R}$
 \newline  {\bf Output: } Subset of the union of user sets $U \subseteq \U = \cup_{u=1}^n S_u$, noisy weight vector $\tilde{w}_{ext}$
}
\begin{algorithmic}[1]
\label{alg:generic}
    \STATE Select $\sigma$ corresponding to the Gaussian Mechanism (\cref{prop:gaussian}) for $(\eps, \delta/2)$-DP with $\Delta_2 = 1$.
    \STATE Set $\rho \gets \max_{t \in [\Delta_0]} h(t) + \sigma \Phi^{-1}\left(\left(1 - \frac{\delta}{2}\right)^{1/t}\right)$ 
    \FORALL{$u \in [n]$} 
        \IF{$|S_u| \geq \Delta_0$} 
            \STATE Randomly subsample $S_u$ to $\Delta_0$ items. \COMMENT{Cap user degrees.}
        \ENDIF
    \ENDFOR
    \STATE $w \gets \weightalgo(\sets)$ \COMMENT{Weights on items in $\U$}
    \STATE $\tilde{w}(i) \gets w(i) + \mathcal{N}(0, \sigma^2 I)$ \COMMENT{Add noise}
    \STATE $U \gets \{i \in \U : \tilde{w}(i) \geq \rho\}$ \COMMENT{Apply threshold.}
    \STATE $\tilde{w}_{ext}(i) \gets 
        \begin{cases}
        \tilde{w}(i) \; &\text{ if } i \in \U \\
        \mathcal{N}(0, \sigma^2 I) \; &\text{ if } i \in \Sigma \setminus \U
        \end{cases}$
        \\ \COMMENT{The $i \in \Sigma \setminus \U$ part is only for privacy analysis (we only ever query this vector on $i \in \U$).}
    \RETURN $U, \tilde{w}_{ext}$ 
\end{algorithmic}
\end{algorithm}

In this section, we formalize the weighting-based meta-algorithm used in prior solutions to the differentially private partition selection problem~\cite{korolova2009releasing, gopi2020dpunion, carvalho2022incorporatingitem, swanberg2023dpsips}. Our algorithm \ouralgo{} also falls within this high-level approach with a novel weighting algorithm that is both adaptive and scalable to massive data.
We alter the presentation of the algorithm from prior work in a subtle, but important, way by having the algorithm release a noisy weight vector $\tilde{w}_{ext}$ in addition to the normal set of items $U$.
This allows us to develop a two-round version of our algorithm \ouralgotworounds{} which queries noisy weights from the first round to give improved performance in the second round, leading to signficant empirical benefit.

The weight and threshold meta-algorithm is given in \cref{alg:generic}. Input is a set of user sets $\sets=\{(u, S_u)\}_{u \in [n]}$, privacy parameters $\eps$ and $\delta$, a maximum degree cap $\Delta_0$, and a weighting algorithm \weightalgo{} (which can itself take some optional input parameters), and a function $h: \mathbb{N} \rightarrow \mathbb{R}$ which describes the sensitivity of \weightalgo{}.

First, each user's set is randomly subsampled so that the size of each resulting set is at most $\Delta_0$ (the necessity of this step will be further explicated).
Then, the \weightalgo{} takes in the cardinality-capped sets and produces a set of weights over all items in the union.
Independent Gaussian noise with standard deviation $\sigma$ is added to each coordinate of the weights, and items with weight above a certain threshold $\rho$ are output.
By construction, this algorithm will only ever output items which belong to the true union,  $U \subseteq \U$, with the size of the output depending on the number of items with noised weight above the threshold.

For the sake of analysis (and not the implementation of the algorithm), we diverge from prior work to return a vector $\tilde{w}_{ext}$ of noisy weights over the entire universe $\Sigma$.
This vector is implicitly used in the proof of privacy for releasing the set of items $U$, but it is never materialized as $|\Sigma|$ is unbounded. 
Within our algorithms, we will ensure that we only ever query entries of this vector which belong to $\U$, so we only ever have to materialize those entries.
Note, however, that it would \emph{not} be private to release $\tilde{w}$ as the output of a final algorithm as the domain of that vector is exactly the true union of the users' sets.

The privacy of this algorithm depends on certain ``sensitivity'' properties of \weightalgo{} as well as our choice of $\sigma$ and $\rho$. Consider any pair of neighboring inputs $\sets$ and $\sets' = \sets \cup \{(v, S_v)\}$, let $\U$ and $\U'$ be the corresponding unions, and let $w$ and $w'$ be the item weights assigned by \weightalgo{} on the two inputs, respectively.

\begin{definition}\label{def:l2-sensitivity}
The \emph{$\ell_2$ sensitivity} of a weighting algorithm is defined as the smallest value $\Delta_2$ such that,
\begin{equation*}
    \Delta_2 \geq \sqrt{\sum_{i \in \U} \left(w'(i) - w(i)\right)^2 + \sum_{i \in \U' \setminus \U} w'(i)^2}.
\end{equation*}
\end{definition}

Given bounded $\ell_2$ sensitivity, choosing the scale of noise $\sigma$ appropriately for the Gaussian mechanism in~\cref{prop:gaussian} ensures that outputting the noised weights on items in $\U$ satisfies $\left(\eps, \frac{\delta}{2} \right)$-DP.
So if we knew $\U$, then the output of the algorithm after thresholding would be private via post-processing.

However, knowledge of the union $\U$ is exactly the problem we want to solve. The challenge is that there may be items in $\U'$ which do not appear in $\U$. Let $T = \U' \setminus \U$ be these ``novel'' items with $t = |T|$. As long as the probability that any of these items are output by the algorithm is at most $\frac{\delta}{2}$, $(\eps, \delta)$-DP will be maintained. Consider a single item $i \in T$ which has zero probability of being output by a weight and threshold algorithm run on $\sets$ but is given some weight $w'(i)$ when \weightalgo{} is run on $\sets'$. The item will be output only if after adding the Gaussian noise with standard deviation $\sigma$, the noised weight exceeds $\rho$. The probability that any item in $T$ is output follows from a union bound. In order to union bound only over finitely many events, we rely on the fact that $t \leq \Delta_0$; this is why the cardinalities must be capped.
This motivates the second important sensitivity measure of \weightalgo{}.

\begin{definition}\label{def:linf-sensitivity}
The \emph{novel $\ell_\infty$ sensitivity} of a weighting algorithm is parameterized by the number $t = |T|$ of items which are unique to the new user,  and is defined as the smallest value $\Delta_\infty(t)$ such that for all possible inputs $\{S_u\}_{u=1}^n$ and new user sets $S_v$,
\begin{equation*}
    \Delta_\infty(t) \geq \max_{i \in T} w'(i).
\end{equation*}
\end{definition}
Then, the calculation of $\rho$ to obtain $(\eps, \delta)$-DP is obtained based on the novel $\ell_\infty$ sensitivity, $\delta$, $\sigma$, and $\Delta_0$. This is formalized in the following theorem whose proof is given in \cref{sec:privacy}.

\begin{theorem}\label{thm:generic-privacy}
Let $\sets, (\eps, \delta), \Delta_0, \weightalgo{}, h$  be inputs to \cref{alg:generic}. 
If \weightalgo{} has bounded $\ell_2$ and novel $\ell_\infty$ sensitivities
\[
    \Delta_2 \leq 1 \; \text{ and } \; \Delta_\infty(t) \leq h(t),
\]
then releasing $U, \tilde{w}_{ext}$ satisfies $(\eps, \delta)$-DP.
\end{theorem}


\section{Adaptive Weighting Algorithms}\label{sec:algo}

\begin{algorithm}[h]
\caption{\userweights{}$(S_u, b, \bmin, \bmax)$
\newline{\bf Input: }
 User set $S_u \subseteq \U$,
 biases $b: \U \to [0,1]$,
 minimum bias $\bmin \in [0.5,1]$, maximum bias $\bmax \in [1,\infty]$
 \newline  {\bf Output: } $w_b: S_u \to \mathbb{R}$ weighting of the items
}
\begin{algorithmic}[1]
\label{alg:user-biased}
\STATE Initialize weight vector $w_b$ with zeros
\STATE $S_{biased} = \{i \in S_u: b(i) < 1\}$
\STATE $S_{unbiased} = S_u \setminus S_{biased}$
\STATE $w_b(i) \gets \frac{\max\{\bmin, b(i)\}}{\sqrt{|S_u|}}$ for $i \in S_{biased}$
\COMMENT{Set biased weights, respecting min bias}
\STATE $w_b(i) \gets \min\bc{\frac{\bmax}{\sqrt{|S_u|}}, \sqrt{\frac{1-\sum_{i \in S_{biased}} w_b(i)^2}{|S_{unbiased}|}}}$ for $i \in S_{unbiased}$
\COMMENT{Allocate remaining $\ell_2$ budget, respecting max bias}
\WHILE{$\sum_{i \in S_u} w_b(i)^2 < 1$}
    \STATE $S_{small} \gets \bc{i \in S_u : w_b(i) < \frac{1}{\sqrt{|S_u|}}}$
    \STATE $C \gets \min\bc{
    \frac{\bmax/\sqrt{|S_u|}}{\max_{i \in S_{small}} w_b(i)},
    \sqrt{1 + \frac{1 - \sum_{i \in S_u} w_b(i)^2}{\sum_{i \in S_{small}} w_b(i)^2}}}$
    \STATE $w_b(i) \gets C \cdot w_b(i)$ for $i \in S_{small}$
    \COMMENT{Increase small weights  using remaining $\ell_2$ budget, respecting max bias}
\ENDWHILE
\RETURN $w_b$
\end{algorithmic}
\end{algorithm}

\begin{algorithm*}[ht]
\footnotesize
\caption{\ouralgo$(\sets, \tau, \dmax, b, \bmin, \bmax)$
 \newline{\bf Input: }
 User sets $\sets = \{(u, S_u)\}_{u \in [n]}$,
 adaptive threshold $\tau \geq 0$,
 maximum adaptive degree $\dmax > 1$,
 biases $b: \U \to \mathbb{R}$,
 minimum bias $\bmin \in [0.5,1]$, maximum bias $\bmax \in [1,\infty]$
 \newline  {\bf Output: } $w: \U \to \mathbb{R}$ weighting of the items
}
\begin{algorithmic}[1]
\label{alg:max-deg}
\STATE Initialize weight vectors $w, w_{init}, w_{trunc}, w_{reroute}$ with zeros.
\STATE Set reroute discount factor $\alpha = \bmin - \frac{1}{2\sqrt{\dmax}}$.
\STATE $I_{adapt} = \{u \in [n]: \ceil{\frac{1}{(\bmin)^2}} \leq |S_u| \leq \dmax\}$ \COMMENT{Only users with certain degrees act adaptively.}
\FORALL{$u \in I_{adapt}$}
    \STATE $w_{init}(i) \pluseq 1/|S_u| \qquad \forall i \in S_u$ \COMMENT{Initial $\ell_1$ sensitivity bounded weights.}
\ENDFOR
\STATE $r(i) \gets \min\left\{0, \frac{w_{init}(i) - \tau}{w_{init}(i)}\right\}$ for $i \in \U$ \COMMENT{Fraction of weight that exceeds the threshold.}
\STATE 
\STATE $w_{trunc}(i) \gets \min\{w_{init}(i), \tau\}$ for $i \in \U$  \COMMENT{Truncate weights above threshold.}
\FORALL{$u \in I_{adapt}$}
    \STATE $e_u \gets \left(1/|S_u|\right) \sum_{i \in S_u} r(i)$ \COMMENT{Excess weight returns to each user proportional to their contribution.}
    \STATE $w_{reroute}(i) \pluseq \alpha e_u / d_{max}$ for $i \in S_u$ \COMMENT{Reroute excess to items, discounted by $\alpha/\dmax$.}
\ENDFOR
\STATE $w \gets w_{trunc} + w_{reroute}$ \COMMENT{Total $\ell_1$ bounded adaptive weights.}
\STATE $w_b^u \gets \userweights{}(S_u, b, \bmin, \bmax)  \quad \forall u \in [n]$ \COMMENT{See \cref{alg:user-biased}.}
\FORALL{$u \in I_{adapt}$}
    \STATE $w(i) \pluseq w_b^u(i) - 1/|S_u|$ for $i \in S_u$ \COMMENT{Add $\ell_2$ bounded weight and subtract initial weights.}
\ENDFOR
\FORALL{$u \in [n] \setminus I_{adapt}$}
    \STATE $w(i) \pluseq w_b^u(i)$ for $i \in S_u$ \COMMENT{Add $\ell_2$ bounded weight for non-adaptive items.}
\ENDFOR
\RETURN $w$
\end{algorithmic}
\end{algorithm*}

Our main result is an \emph{adaptive} weighting algorithm \ouralgolong{} (\ouralgo{}) which is amenable to parallel implementations and has the exact same $\ell_2$ and novel $\ell_\infty$ sensitivities as \basicalgo{}. Therefore, within the weight and threshold meta-algorithm, both algorithms utilize the same noise $\sigma$ and threshold $\rho$ to maintain privacy.
Our algorithm improves upon \basicalgo{} by reallocating weight from items far above the threshold to other items.

We present the full algorithm in \cref{alg:max-deg} with the \userweights{} subroutine given in \cref{alg:user-biased}.
For simplicity, we will first describe the ``unbiased''  version of our algorithm where $b, \bmin, \bmax$ are set to ones and $\userweights{}(S_u, b, \bmin, \bmax)$ is a vector of weights over all items with $\nicefrac{1}{\sqrt{|S_u|}}$ for every $i \in S_u$ and zeros in other coordinates.
The algorithm takes two additional parameters: a maximum adaptive degree $\dmax \in (1, \Delta_0]$ and an adaptive threshold $\tau = \rho + \beta \sigma$ for a free parameter $\beta \geq 0$. Users with set cardinalities greater than $\dmax$ are set aside and contribute basic uniform weights to their items at the end of the algorithm. The rest of the users participate in adaptive reweighting. We start from a uniform weighting where each user sends $\nicefrac{1}{|S_u|}$ weight to each of their items. Items have their weights truncated to $\tau$ and any excess weight is sent back to the users proportional to the amount they contributed. Users then reroute a carefully chosen fraction (depending on $\dmax$) of this excess weight across their items. Finally, each user adds $\nicefrac{1}{\sqrt{|S_u|}} - \nicefrac{1}{|S_u|}$ to the weight of each of their items. 

Each of these stages requires linear work in the size of the input, i.e. the sum of the sizes of the users sets.
Furthermore, each stage is straightforward to implement within a parallel framework. As there are a constant number of stages, the algorithm can be implemented with total linear work and constant number of rounds.

\subsection{\ouralgotworounds{}: Biased Weights in Multiple Rounds}

\begin{algorithm}[ht]
\caption{\ouralgotworounds{}$(\sets, (\eps_1, \delta_1), (\eps_2, \delta_2), \Delta_0, \dmax, \beta, C_{lb}, C_{ub}, \bmin, \bmax)$
 \newline  {\bf Input: } User sets $\sets=\{(u, S_u)\}_{u\in [n]}$,
 privacy parameters $(\eps_1, \delta_1)$ and $(\eps_2, \delta_2)$,
 degree cap $\Delta_0$,
 maximum adaptive degree $\dmax$,
 adaptive threshold excess parameter $\beta$,
 lower bound constant $C_{lb}$,
 upper bound constant $C_{ub}$,
 minimum bias $\bmin$,
 maximum bias $\bmax$
 \newline  {\bf Output: } Subset of the union of user sets $\U = \cup_{u=1}^n S_u$
}
\begin{algorithmic}[1]
\label{alg:max-deg-two-round}
\STATE For $u \in [n]$, cap $S_u$ to at most $\Delta_0$ items by random subsampling
\STATE Select $\sigma_r$ corresponding to the Gaussian Mechanism (\cref{prop:gaussian}) for $(\eps_r, \delta_r/2)$-DP with $\Delta_2 = 1$ for $r \in \{1,2\}$.
\STATE \textbf{Round 1}
\STATE Set threshold $\rho_1 = \max_{t \in [\Delta_0]} \frac{1}{\sqrt{t}} + \sigma_1 \Phi^{-1}\left(\left(1 - \frac{\delta}{2}\right)^{1/t}\right)$
\STATE $w_1 \gets \ouralgo(\sets, \rho_1 + \beta \sigma_1, \dmax, \overrightarrow{1}, 1, 1)$ \COMMENT{Compute \ouralgo{} (\cref{alg:max-deg}) weights in the first round.}
\STATE $\tilde{w_1} \gets w_1 + \mathcal{N}(0, \sigma_1^2 I)$ \COMMENT{Add noise}
\STATE $U_1 \gets \{i \in \U : \tilde{w}_1(i) \geq \rho_1\}$ \COMMENT{Apply threshold}
\STATE \textbf{Round 2}
\STATE Set threshold $\rho_2 = \max_{t \in [\Delta_0]} \frac{\bmax}{\sqrt{t}} + \sigma_2 \Phi^{-1}\left(\left(1 - \frac{\delta}{2}\right)^{1/t}\right)$
\STATE $\tilde{w}_{lb} \gets \max\bc{0, \tilde{w}_1 - C_{lb} \cdot \sigma_1}$
\COMMENT{Weight lower bound from Round 1}
\STATE $\tilde{w}_{ub} \gets \tilde{w}_1 + C_{ub} \cdot \sigma_1$
\COMMENT{Weight upper bound from Round 1}
\STATE $U_{low} \gets \bc{i \in \U: \tilde{w}_{ub} < \rho_2}$
\STATE $S_u \gets S_u \setminus \p{U_1 \cup U_{low}}$ for $u \in [n]$
\COMMENT{Remove items found in Round 1 or with a small upper bound on the weight}
\STATE $b \gets \min\bc{1, \frac{\rho_2}{\tilde{w}_{lb}}}$
\COMMENT{Bias weights to not overshoot threshold}
\STATE $w_2 \gets \ouralgo{}(\sets, \rho_2 + \beta \sigma_2, \dmax, b, \bmin, \bmax)$
\COMMENT{Compute \ouralgo{} (\cref{alg:max-deg}) in the second round}
\STATE $\tilde{w}_2 \gets w_2 + \mathcal{N}(0, \sigma_2^2 I)$ \COMMENT{Add noise}
\STATE $U_2 \gets \{i \in \U : \tilde{w}_2(i) \geq \rho_2\}$ \COMMENT{Apply threshold}
\RETURN $U_1 \cup U_2$
\end{algorithmic}
\end{algorithm}

This unbiased version of \ouralgo{} directly improves on the basic algorithm.
We further optimize our algorithm by refining an idea from the prior work of DP-SIPS~\cite{swanberg2023dpsips}.
In that work, the privacy budget is split across multiple rounds with \basicalgo{} used in each round.
In each round, items found in previous rounds are removed from the users' sets, so that in early rounds, easy-to-output (loosely speaking, high frequency) items are output, with more weight being allocated to harder-to-output items in future rounds.
The privacy of this approach follows from post-processing: we can freely use the differentially private output $U$ from early rounds to remove items in later rounds.

We propose \ouralgotworoundslong{} (\ouralgotworounds{}) which as a starting point runs \ouralgo{} in two rounds, splitting the privacy budget as in DP-SIPS. As \ouralgo{} stochastically dominates \basicalgo{}, this provides a drop-in improvement.
Our key insight comes from the modified meta-algorithm we present in \cref{sec:meta-algo} which also outputs the vector of noisy weights $\tilde{w}_{ext}$.
As long as the \weightalgo{} maintains bounded sensitivity, we are free to query the noisy weights from prior rounds when constructing weights in future rounds.

We leverage this by running in two rounds with the full pseudocode given in \cref{alg:max-deg-two-round}.
In the first round, we run the unbiased version of \ouralgo{} described above to produce outputs $U_1 = U$ as well as query access to $\tilde{w}_{ext}$.
We will only ever query items in $\U$, so we maintain $\tilde{w}_1$ which is $\tilde{w}_{ext}$ restricted to $\U$ without ever materializing $\tilde{w}_{ext}$. Importantly though, we never release $\tilde{w}_1$ as a final output.

In the second round, we make three preprocessing steps before running \ouralgo{}. Let $\sigma_1$ be the standard deviation of the noise in the first round and $\rho_2$ be the threshold in the second round. For parameters $C_{lb}, C_{ub} \geq 0$, Let $\tilde{w}_{lb} = \tilde{w}_1 - C_{lb} \cdot \sigma_1$ and $\tilde{w}_{ub} = \tilde{w}_1 + C_{ub} \cdot \sigma_1$ be lower and upper confidence bounds on the true item weights $w_1$ in the first round, respectively.
\begin{enumerate}[label=(\alph*)]
    \item (DP-SIPS) We remove items from users' sets which belong to $U_1$.
    \item (Ours) We remove items $i$ from users' sets which have weight significantly below the threshold where $\tilde{w}_{ub}(i) < \rho_2$. If $\tilde{w}_{ub}(i)$ is very small, we have little chance of outputting the item in the second round and would rather not waste weight on those items. This is particularly relevant for long-tailed distributions we often see in practice where there are many elements which appear in only one or a few users' sets.
    \item (Ours) For items with $\tilde{w}_{lb} \geq \rho_2$, we assign these items biases $b(i) = \nicefrac{\rho_2}{\tilde{w}_{lb}}(i)$. Via \userweights{} (\cref{alg:user-biased}), we (loosely) try to have each user contribute a $b(i)$ fraction of their normal $\nicefrac{1}{\sqrt{|S_u|}}$ weight while increasing the weights on unbiased items.
    As the lower bound on these item weights is very large, we do not need to spend as much of our $\ell_2$ budget on these items.
    For technical reasons, in order to preserve the overall sensitivity of \ouralgo{}, we must enforce minimum and maximum bias parameters $\bmin \in [0.5, 1]$ and $\bmax \in [1, \infty)$.
    The weights returned by \userweights{} are in the interval $\br{\nicefrac{\bmin}{\sqrt{|S_u|}}, \nicefrac{\bmax}{\sqrt{|S_u}}}$ and have an $\ell_2$ norm of 1.
\end{enumerate}




\section{Privacy Analysis}\label{sec:privacy}
\subsection{Meta-Algorithm}
We start by proving the privacy of the meta-algorithm described in \cref{sec:meta-algo}.
\begin{proof}[Proof of \cref{thm:generic-privacy}]\label{proof:generic-privacy}
Let $w_{ext}: \Sigma \rightarrow \mathbb{R}$ be an extension of the weight vector $w$ returned by \weightalgo{} where
\begin{equation*}
w_{ext}(i) = 
\begin{cases}
    w(i) \; &\text{ if } i \in \U \\
    0 \; &\text{ if } i \in \Sigma \setminus \U
\end{cases}.
\end{equation*}
Note that $\tilde{w}_{ext}$ is exactly the result of applying the Gaussian Mechanism to $w_{ext}$.
Furthermore, the $\ell_2$ sensitivity (\cref{def:l2-sensitivity}) of computing $w_{ext}$ is the same as that of $w$ as items outside of $\U'$ do not contribute to the sensitivity as they are $0$ regardless of whether the input is $\sets$ or $\sets'$.
By the choice of $\sigma$ according to \cref{prop:gaussian} with $\Delta_2 \leq 1$, releasing $\tilde{w}_{ext}$ is $\p{\eps, \frac{\delta}{2}}$-DP.

The privacy of releasing $U$ depends on our choice of the threshold $\rho$.
We will first show that the probability that any of $t$ i.i.d.\ draws from a Gaussian random variable $\mathcal{N}(0, \sigma^2)$ exceeds $\sigma \Phi^{-1}\p{\p{1 - \frac{\delta}{2}}^{1/t}}$ is exactly $\frac{\delta}{2}$.
Let $A$ be the bad event, $Y \sim \mathcal{N}(0, \sigma^2)$, and $Z \sim \mathcal{N}(0, 1)$:
\begin{align*}
   \Pr(A) &= 1 - \Pr\p{Y \leq \sigma \Phi^{-1}\p{\p{1 - \frac{\delta}{2}}^{1/t}}}^t \\
   &= 1 - \Pr\p{Z \leq \Phi^{-1}\p{\p{1 - \frac{\delta}{2}}^{1/t}}}^t \\
   &= 1 - \Phi\p{\Phi^{-1}\p{\p{1 - \frac{\delta}{2}}^{1/t}}}^t \\
   &= 1 - \p{1 - \frac{\delta}{2}}^{t/t} \\
   &= \frac{\delta}{2}.
\end{align*}

By the condition that $h(t)$ is an upper bound on $\Delta(t)$, the choice of $\rho$ implies that, no matter how many items are novel (unique to the new user in a neighboring dataset), the probability that any of them belong to $U$ is at most $\frac{\delta}{2}$.
Conditioned on the release of $\tilde{w}_{ext}$, releasing $U$ is $\p{0, \frac{\delta}{2}}$-DP.
By basic composition, the overall release is $(\eps, \delta)$-DP, as required.
\end{proof}


\subsection{Adaptive Algorithms}
We now prove the privacy of our main algorithm by bounding its $\ell_2$ and novel $\ell_\infty$ sensitivities.
\begin{theorem}[Privacy of \ouralgotworounds{}]\label{thm:dp-two-rounds}
\cref{alg:max-deg-two-round} is $(\eps_1 + \eps_2, \delta_1 + \delta_2)$-DP.
\end{theorem}

The rest of this section will be devoted to proving this theorem. We first state as a corollary that \ouralgo{} run in a single round without biases is also private.

\begin{corollary}[Privacy of \ouralgo{}]\label{thm:dp-one-round}
Releasing the output $U$ from \cref{alg:generic} run with unbiased \ouralgo{} (\cref{alg:max-deg}) as the weighting algorithm and with $h(t) = \frac{1}{\sqrt{t}}$ is $(\eps, \delta)$-DP.
\end{corollary}
\begin{proof}
This follows directly from \cref{thm:dp-two-rounds} by setting $(\epsilon_1, \delta_1) = (\epsilon, \delta)$ and $(\epsilon_2, \delta_2) = (0,0)$.
\end{proof}

To prove privacy, consider a weight vector $w$ returned by \cref{alg:max-deg} for an input $\sets = \{(u, S_u)\}_{u=1}^n, \tau, \dmax, b, \bmin, \bmax$ and the output $w'$ for an input $\sets' = \sets \cup \{(v, S_v)\}, \tau, \dmax, b', \bmin, \bmax$ which includes a new user $v$ not in the original input.
Let $d_v$ be the degree of the new user.
Let $T = S_v \setminus \cup_{u=1}^n S_u$ be the subset of items which appear only in $S_v$ and not in any of the original user sets, and let $t = |T|$. The vectors of biases $b$ and $b'$ are defined on the set of items $\cup_{u=1}^n S_u$ and 
$\left(\cup_{u=1}^n S_u \right) \cup T$ respectively. 
We remind the notation used in \cref{alg:max-deg} for vector $w^v_b$ to denote the biased weight vector the new user $v$ computes by calling \cref{alg:user-biased}, i.e. \userweights.



\begin{lemma}[Novel $\ell_\infty$ sensitivity with biases]\label{lem:linf-sensitivity-biased}
Assume that the adaptive threshold is $\tau \geq 1$, and the maximum adaptive degree is $\dmax \geq 4$. Then, \cref{alg:max-deg} has novel $\ell_\infty$ sensitivity bounded by
\[
    \Delta_\infty(t) \leq \frac{\bmax}{\sqrt{t}}.
\]
where $t = |T|$ is the cardinality of $T$, the set of novel items.
\end{lemma}
\begin{proof}
Unpacking \cref{def:linf-sensitivity}, it suffices to show that
\[
    \max_{i \in T} w'(i) \leq \frac{\bmax}{\sqrt{t}}.
\]
Note that this is trivially true if $d_v = |S_v| > \dmax$ or $d_v < \ceil{\frac{1}{(\bmin)^2}}$ since $v$ does not participate in adaptivity due to its too low or high degree. We will proceed by assuming this is not the case.

Consider the final weight of an item $i$ in the set of novel items $T \subseteq S_v$:
\begin{equation*}
    w'(i) = w'_{trunc}(i) + w'_{reroute}(i) + w^v_b(i) - w'_{init}(i). 
\end{equation*}

Note that for all novel items $i \in T$, $w'_{init}(i) = \frac{1}{d_v} \leq \tau$ as $v$ is the sole contributor to the weight of item $i$.
Therefore, no weight is truncated or rerouted from these items: $w'_{trunc}(i) = w'_{init}(i) =  \frac{1}{d_v}$.
Expanding the definition of rerouted weight,
\begin{align*}
    w'_{reroute}(i) &= \frac{\alpha e'_v}{\dmax} \\
    &= \frac{\alpha}{\dmax d_v} \sum_{j \in S_v \setminus T} r'(j) \\
    &= \frac{\alpha}{\dmax d_v} \sum_{j \in S_v \setminus T} \frac{w'_{init}(j) - w'_{trunc}(j)}{w'_{init}(j)} \\
    &\leq \frac{\alpha}{\dmax d_v} \p{d_v - t}\\
    &= \frac{\alpha}{\dmax}\p{1 - \frac{t}{d_v}}.
\end{align*}
Finally, note that $w_b^v(i) \leq \frac{\bmax}{\sqrt{d_v}}$ by construction (this is the meaning of $\bmax$).

We can bound $w'(i)$ as
\begin{align}
    w'(i) &= w'_{trunc}(i) + w'_{reroute}(i) + w^v_b(i) - w'_{init}(i) \notag \\
    &\leq \frac{1}{d_v} + \frac{\alpha}{\dmax}\p{1 - \frac{t}{d_v}} + \frac{\bmax}{\sqrt{d_v}} - \frac{1}{d_v} \notag \\
    &= \frac{\alpha}{\dmax}\p{1 - \frac{t}{d_v}} + \frac{\bmax}{\sqrt{d_v}}. \label{eq:linf-upper-bound}
\end{align}
In the rest of the proof, we will show that the upper bound \cref{eq:linf-upper-bound} is maximized when $d_v = t$, i.e., when the first term is zero.
Recall that $t \leq d_v \leq \dmax$.
Consider the partial derivative with respect to $d_v$:
\begin{equation*}
    \frac{\partial}{\partial d_v} \p{\frac{\alpha}{\dmax}\p{1 - \frac{t}{d_v}} + \frac{\bmax}{\sqrt{d_v}}}
    = \frac{\alpha}{\dmax} \frac{t}{d_v^2} - \frac{\bmax}{2d_v^{3/2}}.
\end{equation*}
Consider the condition of the derivative being non-positive:
\begin{align*}
    &0 \geq \frac{\alpha}{\dmax} \frac{t}{d_v^2} - \frac{\bmax}{2d_v^{3/2}} \\
    \iff &\bmax \geq \frac{2\alpha t}{\dmax \sqrt{d_v}} \\
    \Longleftarrow\; &\bmax \geq \frac{2\alpha\sqrt{t}}{\dmax} \\
    \Longleftarrow\; &\bmax \geq \frac{2\alpha}{\sqrt{\dmax}}.
\end{align*}
The final condition holds as $\alpha \leq 1$ and by the assumption that $\dmax \geq 4$. We note that $\bmax$ is always set to be at least $1$.
As the derivative is non-positive, the right side of \cref{eq:linf-upper-bound} is maximized when $d_v$ is minimized at $d_v=t$.
Then, $\Delta_\infty(t) \leq \frac{\bmax}{\sqrt{t}}$, as required.
\end{proof}

Following we state some properties of the biased weights $w_b$ which will be helpful in the proof.
\begin{lemma}\label{lem:l2-bound-user-weights}
Let $S_u, b, \bmin, \bmax$ be valid inputs to \cref{alg:user-biased}, and let $w_b$ be the weight vector returned by the algorithm.
Let $d = |S_u|$.
Then, the following hold:
\begin{itemize}
    \item $w_b(i) = 0$ for all $i \in \U \setminus S_u$
    \item $\frac{\bmin}{\sqrt{d}} \leq w_b(i) \leq \frac{\bmax}{\sqrt{d}}$ for all $i \in S_u$
    \item $\|w_b\|_2 \leq 1$
\end{itemize}
\end{lemma}

\begin{proof}
The first claim holds as the weight vector is initialized with zeros and only indices $i \in S_u$ are updated by the algorithm.

To simplify the notation, we define $d = |S_u| = |S_{biased}| + |S_{unbiased}|$.
As $b(i) \leq 1$ for all $i \in \U$, the initial weights given to items in $S_{biased}$ are between $\frac{\bmin}{\sqrt{d}}$ and $\frac{1}{\sqrt{d}}$. We also know that $\frac{1}{\sqrt{d}} \leq \frac{\bmax}{\sqrt{d}}$.
The sum of squares of these weights will thus be between $\frac{\bmin^2 |S_{biased}|}{d}$ and $\frac{|S_{biased}|}{d}$.
Call this value $k$.

Weights of items in $S_{unbiased}$ are given by
\begin{equation*}
    \min\bc{\frac{\bmax}{\sqrt{d}}, \sqrt{\frac{1 - k}{|S_{unbiased}|}}}
\end{equation*}
We show that the minimum of these two terms is at least $\frac{1}{\sqrt{d}} \geq \frac{\bmin}{\sqrt{d}}$. The first term is at least $\frac{1}{\sqrt{d}}$ since $\bmax \geq 1$. To observe the same for the second term, one should   plugg in the upper bound of $k \leq \frac{|S_{biased}|}{d}$. 
By construction, the weights are upper bounded by $\frac{\bmax}{\sqrt{d}}$.
Furthermore, note that the sum of squares of the entire weight vector at this point is upper bounded by $1$. In particular, it is equal to $1$ for the second term of the minimization:
\begin{equation*}
    k + \sum_{i \in S_{biased}} \p{\frac{1 - k}{|S_{unbiased}|}}
    = k + (1 - k)
    = 1.
\end{equation*}

In the remainder of the algorithm, sum of the weights of items in $S_{small} \subseteq S_{unbiased}$ may increase if the $\ell_2$ norm of the weight vector is strictly less than $1$.
Consider the weights after any such update by a multiplicative factor $C$ defined as
\begin{equation*}
    C = \min\bc{
    \frac{\bmax/\sqrt{d}}{\max_{i \in S_{small}} w_b(i)},
    \sqrt{1 + \frac{1 - \sum_{i \in S_u} w_b(i)^2}{\sum_{i \in S_{small}} w_b(i)^2}}}.
\end{equation*}
Note that $C > 1$ by definition of $S_{small}$ and the stopping criteria of the while loop. Therefore, none of the final weights will be less than $\frac{\bmin}{\sqrt{d}}$.
Consider the first case of the minimization.
Any updated weight $C \cdot w_b(i)$ for $i \in S_{small}$ will be at most $\frac{\bmax}{\sqrt{d}}$ as
\begin{equation*}
    \frac{\bmax/\sqrt{d}}{\max_{j \in S_{small}} w_b(j)} \cdot w_b(i)
    \leq 
    \max_{i^* \in S_{small}} \frac{\bmax/\sqrt{d}}{w_b(i^*)} w_b(i^*)
    = \bmax/\sqrt{d}.
\end{equation*}
Now, consider the second case of the maximization.
Then, the squared $\ell_2$ norm of the weight vector will be
\begin{equation*}
    \sum_{i \in S_u \setminus S_{small}} w_b(i)^2
    + \sum_{i \in S_{small}} \p{1 + \frac{1 - \sum_{j \in S_u} w_b(j)^2}{\sum_{j \in S_{small}} w_b(j)^2}} w_b(i)^2
    = \p{\sum_{i \in S_u} w_b(i)^2} + \p{1 - \sum_{j \in S_u} w_b(j)^2}
    = 1.
\end{equation*}
As $C$ is taken to be the minimum of these two values, the final weight vector will satisfy all of the required bounds.
\end{proof}

We will prove a useful fact that \cref{alg:max-deg} is monotone in the sense that weights when run on $\sets'$ will only increase compared to when run on only $\sets$. We apply this proposition in upper bounding the $\ell_2$ sensitivity of \cref{alg:max-deg} in \cref{lem:l2-sensitivity-biased-weights}.
\begin{proposition}[Monotonicity]\label{lem:monotone}
For all $i \in \cup_{u \in \{1, \ldots, n, v\}} S_u$,
\begin{align*}
    w'_{init}(i) &\geq w_{init}(i) \\
    w'_{trunc}(i) &\geq w_{trunc}(i) \\
    w'_{reroute}(i) &\geq w_{reroute}(i).
\end{align*}
\end{proposition}
\begin{proof}
Fix any item $i$.
As all increments to the initial $\ell_1$ bounded weights are positive and non-adaptive,
\[
    w'_{init}(i) \geq w_{init}(i).
\]
In fact, the two weights are either equal or differ by a factor of $1/d_v$ depending on whether $i \in S_v$.
The calculations of the fraction of excess weight that exceeds the threshold, the truncated weights, the excess weight returned to each user, and the rerouted weights are all monotonically non-decreasing with the initial weights.
Therefore,
\[
    w'_{trunc}(i) \geq w_{trunc}
\]
and
\[
    w'_{reroute}(i) \geq w_{reroute}.
\]
\end{proof}


\begin{lemma}[$\ell_2$ sensitivity with biased weights]\label{lem:l2-sensitivity-biased-weights}
\cref{alg:max-deg} has $\ell_2$-sensitivity upper bounded by $1$.
\end{lemma}

\begin{proof}
Note that this is trivially true by $\cref{lem:l2-bound-user-weights}$ if $|S_v| = d_v > \dmax$ or $d_v < \ceil{\frac{1}{(\bmin)^2}}$ since $v$ does not participate in adaptivity in this case. 
We will proceed by assuming this is not the case.

Our goal is to bound the $\ell_2$ norm of the difference $\Delta = w' - w$, the difference in final weights with and without the new user $v$. We will use the notation $\Delta_{subscript} = w'_{subscript} - w_{subscript}$.
Note that $\Delta_{reroute} = w'_{reroute} - w_{reroute}$ is the additional rerouted weight after adding user $v$ and let $\Delta_{user} = \Delta - \Delta_{reroute}$ be the rest of the difference.
Note that $\Delta_{user}$ is $d_v$-sparse and only has nonzero entries on $S_v$, the items of the new user.
Our goal will be to bound the $\ell_2$ norms of $\Delta_{reroute}$ and $\Delta_{user}$, thus bounding the $\ell_2$ sensitivity of the algorithm by triangle inequality.

We start by tracking the excess weight created by $v$ which will be useful in bounding the $\ell_2$ norms of both $\Delta_{reroute}$ and $\Delta_{user}$.
It is the total amount of weight added by $v$ to items that exceed the threshold\footnote{If if an item $i$ only exceeds the threshold due to the addition of $v$, we only consider the allocated weight to $i$ above the threshold. }:
\begin{equation}\label{eq:excess}
    \gamma = \|\Delta_{init} - \Delta_{trunc}\|_1
\end{equation}
(Note that this is not the same as $e'_v$ which is the amount of weight from $v$ gets returned to reroute.)

The total amount of weight that is returned to users to reroute is equal to the amount of weight truncated, i.e., the sum of $w_{init} - w_{trunc}$:
\begin{align*}
    \sum_{u=1}^n e_u &= \sum_{u=1}^n (1/|S_u|) \sum_{i \in S_u} r(i) \\
    &= \sum_{u=1}^n (1/|S_u|) \sum_{i \in S_u} \max\left\{0, \frac{w_{init}(i) - \tau}{w_{init}(i)}\right\} \\
    &= \sum_{u=1}^n (1/|S_u|) \sum_{i \in S_u} \frac{w_{init}(i) - w_{trunc}(i)}{w_{init}(i)} \\
    &= \sum_{u=1}^n (1/|S_u|) \sum_{i \in S_u} \frac{w_{init}(i) - w_{trunc}(i)}{\sum_{w : i \in S_w} 1/|S_w|} \\
    &= \sum_{i \in \U} \frac{(w_{init}(i) - w_{trunc}(i)) \sum_{u: i \in S_u} 1/|S_u|}{\sum_{w : i \in S_w} 1/|S_w|} \\
    &= \sum_{i \in \U} w_{init}(i) - w_{trunc}(i).
\end{align*}

For notational parsimony, let $e_v = 0$ (as $v$ does not appear in the original input $\sets$).
Note that $e_u$ is monotonically increasing with $w_{init}$: if any coordinate of the initial weight increases, the excess ratio of any user will never decrease.
With monotonicity of $w_{init}$ from \cref{lem:monotone}, it follows that $w'_{init} - w'_{trunc} \geq w_{init} - w_{trunc}$ since the threshold $\tau$ in the capping formula  $w_{trunc}(i) \gets \min\{w_{init}(i), \tau\}$ stays the same after adding the new user. Consequently, we have:
\begin{align*}
    \gamma
    &= \|\p{w'_{init} - w'_{trunc}} - \p{w_{init} - w_{trunc}}\|_1 \\
    &= \p{\sum_{i \in \U'} w'_{init} - w'_{trunc}} - \p{\sum_{i \in \U} w_{init} - w'_{trunc}} \\
    &= \sum_{u \in \{1, \ldots, n, v\}} e'_u - e_u.
\end{align*}

Now, we will consider $\Delta_{reroute}$.
Recall that $|S_u| \leq \dmax$ for all $u$ participating in adaptivity. For all other users, the terms $e_u$ and $e'_u$ are zero.
\begin{align*}
    \|\Delta_{reroute}\|_1 &= \sum_{u \in \{1, \ldots, n, v\}} \sum_{i \in S_u} (\alpha / \dmax) (e'_u - e_u) \\
    &= \sum_{u \in \{1, \ldots, n, v\}} |S_u| (\alpha / \dmax) (e'_u - e_u) \\
    &\leq \sum_{u \in \{1, \ldots, n, v\}} \alpha(e'_u - e_u) \\ 
    &= \alpha \gamma.
\end{align*}
Furthermore, we can bound the $\ell_\infty$ norm of $\Delta_{reroute}$ as:
\begin{equation*}
    w'_{reroute}(i) - w_{reroute}(i) \leq (\alpha / \dmax) \sum_{u \in \{1, \ldots, n, v\}} e'_u - e_u = \alpha \gamma / \dmax.
\end{equation*}
By H\"older's inequality, 
\begin{equation}\label{eq:reroute-l2}
    \|\Delta_{reroute}\|_2 \leq \sqrt{\|\Delta_{reroute}\|_1 \|\Delta_{reroute}\|_\infty}
    = \sqrt{\alpha^2 \gamma^2 / \dmax} = \frac{\alpha \gamma}{\sqrt{\dmax}}.
\end{equation}

Consider the rest of the difference $\Delta_{user}$. For $i \in S_v$, a single coordinate of $\Delta_{user}$ will be comprised of the sum
\begin{equation*}
    \Delta_{user}(i) = \Delta_{trunc}(i) + w_b^v(i) - 1/d_v.
\end{equation*}
Note that $\Delta_{trunc}(i) \in [0,1/d_v]$, so
\begin{equation*}
    \Delta_{user}(i) \in \left[w_b^v(i) - 1/d_v, w_b^v(i)\right].
\end{equation*}

Furthermore, as $\Delta_{init}(i) = 1/d_v$,
\begin{equation*}
    \|\Delta_{user}\|_1 = 
    \sum_{i \in S_v} w_b^v(i) + \Delta_{trunc}(i) - \Delta_{init}(i)
    = \left(\sum_{i \in S_v} w_b^v(i) \right) - \gamma.
\end{equation*}

Let $x: S_v \rightarrow \mathbb{R}$ and $y: S_v \rightarrow \mathbb{R}$ be two sets of weights over $S_v$ such that $x(i) = w_b^v(i) - 1/d_v$, $y(i) \in [0, 1/d_v]$, and $\sum_{i \in S_v} x(i) + y(i) = \left( \sum_{i \in S_v} w_b^v(i) \right) - \gamma$.
Then,
\begin{equation*}
    \|\Delta_{user}\|_2 \leq \max_y \|x + y\|_2
\end{equation*}
as any valid $\Delta_{user}$ can be expressed as the sum of $x$ and a choice of $y$ satisfying the above constraints.
Note that
\begin{equation*}
    \|y\|_1 = \left(\sum_{i \in S_v} w_b^v(i) \right) - \gamma - \sum_{i \in S_v} x(i)
    = \left(\sum_{i \in S_v} w_b^v(i) \right) - \gamma - \sum_{i \in S_v} w_b^v(i) + 1
    = 1 - \gamma
\end{equation*}
and by H\"older's inequality, 
\begin{equation*}
    \|y\|_2 \leq \sqrt{\|y\|_1 \|y\|_\infty} = \sqrt{\frac{1 - \gamma}{d_v}}.
\end{equation*}
Then,
\begin{align*}
    \|x + y\|_2^2 &= \sum_{i \in S_v} \left(w_b^v(i) - 1/d_v + y(i)\right)^2 \\
    &= \sum_{i \in S_v} w_b^v(i)^2 + \frac{1}{d_v^2} + y(i)^2 - \frac{2 w_b^v(i)}{d_v} + 2 y(i) \cdot w_b^v(i) - \frac{2y(i)}{d_v} \\
    &= 1 + \frac{1}{d_v} + \|y\|_2^2 - \frac{2}{d_v} \cdot \|w_b^v\|_1  + 2 \langle y, w_b^v \rangle - \frac{2}{d_v} \cdot \|y\|_1  \\
    &\leq 1 + \frac{1}{d_v} + \frac{1-\gamma}{d_v} - \frac{2}{d_v} \cdot \|w_b^v\|_1  + 2 \langle y, w_b^v \rangle - \frac{2(1-\gamma)}{d_v}  \\
    &= 1 + \frac{1 + 1-\gamma - 2(1-\gamma)}{d_v}  - \frac{2}{d_v} \cdot \|w_b^v\|_1  + 2 \langle y, w_b^v \rangle   \\
    &= 1  + \frac{\gamma}{d_v} - \frac{2}{d_v} \cdot \|w_b^v\|_1  + 2 \langle y, w_b^v \rangle
\end{align*}

Every $y(i)$ can be written as $1/d_v - z(i)$ for some non-negative residual $z(i)$ with $\sum_{i \in S_v} z(i) = \gamma$. So we continue the above equations as follows: 

\begin{align*}
    \|x + y\|_2^2 
    &\leq  1  + \frac{\gamma}{d_v} - \frac{2}{d_v} \cdot \|w_b^v\|_1  + 2 \langle y, w_b^v \rangle\\
    &=  1  + \frac{\gamma}{d_v} - \frac{2}{d_v} \cdot \|w_b^v\|_1  + \frac{2}{d_v}  \cdot \|w_b^v\|_1 -  2\sum_{i \in S_v} z(i) w_b^v(i) \\
    &=  1  + \frac{\gamma}{d_v}  -  2\sum_{i \in S_v} z(i) w_b^v(i) \\
    &\leq  1  + \frac{\gamma}{d_v}  -  2 \cdot \frac{\bmin}{\sqrt{d_v}} \cdot \sum_{i \in S_v} z(i)  \\
    &= 1 - \frac{2 \bmin \gamma}{\sqrt{d_v}} + \frac{\gamma}{d_v}.
\end{align*}

The last inequality holds because   \cref{lem:l2-bound-user-weights} implies that 
$w_b^v(i) \geq \frac{\bmin}{\sqrt{d_v}}$. We conclude that: 
\[
\|\Delta_{user}\|_2 \leq \sqrt{1 - \frac{2 \bmin \gamma}{\sqrt{d_v}} + \frac{\gamma}{d_v}}.
\]

We can now bound the $\ell_2$-sensitivity of the entire algorithm as
\begin{equation*}
    \|\Delta\|_2 \leq \|\Delta_{reroute}\|_2 + \|\Delta_{user}\|_2 \leq 
    \frac{\alpha \gamma}{\sqrt{\dmax}} + \sqrt{1 - \frac{2 \bmin \cdot \gamma}{\sqrt{d_v}} + \frac{\gamma}{d_v}}.
\end{equation*}
As the expression $-\frac{2 \bmin}{\sqrt{d_v}} + \frac{1}{d_v}$ is increasing for $d_v \geq \frac{1}{\bmin^2}$, the right hand side is maximized with $d_v = \dmax$:

\begin{equation*}
    \|\Delta\|_2 \leq \frac{\alpha \gamma}{\sqrt{\dmax}} + \sqrt{1 - \frac{2 \bmin \cdot \gamma}{\sqrt{\dmax}} + \frac{\gamma}{\dmax}}.
\end{equation*}
Our goal is to choose $\alpha \in [0,1]$ such that the right hand side is upper bounded by $1$.
We note that for $\gamma = 0$, the above inequality proves this desired upper bound.
For other cases, $\gamma \in (0,1]$, it is achieved when,
\begin{equation}
    \alpha \leq \frac{\sqrt{\dmax}}{\gamma} \left(1 - \sqrt{1 - \frac{2 \bmin \cdot \gamma}{\sqrt{\dmax}} + \frac{\gamma}{\dmax}} \right).
\end{equation}
By \cref{lem:gen-limit} with $C=2 \bmin$ and using the restrictions $\dmax > 1$ and $\frac{1}{2} \leq \bmin \leq 1$, it suffices to choose
\begin{equation}\label{eq:alpha}
    \alpha = \bmin - \frac{1}{2\sqrt{\dmax}}.
\end{equation}
\end{proof}

\begin{proof}[Proof of \cref{thm:dp-two-rounds}]
Note that both rounds of \ouralgotworounds{} (\cref{alg:max-deg-two-round}) correspond to running the \selectalgo{} meta-algorithm (\cref{alg:generic}) with privacy parameters $(\eps_1, \delta_1)$ and $(\eps_2, \delta_2)$, respectively.
The only difference is that we only materialize $\tilde{w}_1$ rather than the entire vector $\tilde{w}_{ext}$ from the first round. The functionality of our algorithm would be equivalent if we instead materialized the full vector as we only query weights on items in $\U$ and we never output the vector.
Therefore, we will invoke \cref{thm:generic-privacy} twice and apply basic composition to prove the privacy of \ouralgotworounds{}.
By \cref{thm:generic-privacy}, it suffices to show that the $\ell_2$ and novel $\ell_\infty$ sensitivities of the weight algorithm \ouralgo{} are bounded by $1$ and $\frac{\bmax}{\sqrt{t}}$, respectively.
This follows directly from \cref{lem:linf-sensitivity-biased}, and \cref{lem:l2-sensitivity-biased-weights}.
\end{proof}

\subsection{Technical Lemma}

\begin{proposition}[Taylor expansion of $\sqrt{1+x}$ as $x \to 0$]\label{prop:taylor-sqrt}
\[
\lim_{x \to 0} \sqrt{1 + x} = \sum_{n=0}^\infty \frac{\prod_{k=1}^n \left(\frac{3}{2} - k\right)}{n!} x^n.
\]
\end{proposition}


\begin{lemma}\label{lem:gen-limit}
For a constant $1\leq C \leq 2$, 
consider the following function of $x$ parameterized by an auxiliary variable $y$:
\begin{equation}\label{eq:fxy}
    f(x;y) = \frac{\sqrt{y}}{x}\left(1 - \sqrt{1 - \frac{Cx}{\sqrt{y}} + \frac{x}{y}}\right).
\end{equation}
For any $y > 1$,
\[
    \inf_{x \in (0,1]} f(x;y)
    = \frac{C}{2} - \frac{1}{2\sqrt{y}}.
\]
\end{lemma}
\begin{proof}
To minimize $f$, we will evaluate the function at any stationary points (in terms of $x$) as well as the boundaries $x = 1$ and $x \to 0$.
Consider the derivative
\begin{equation*}\label{eq:derivative}
    \frac{d}{d x}f(x;y) = - \frac{\sqrt{y}}{x^2}\left(1 - \sqrt{1 - \frac{C x}{\sqrt{y}} + \frac{x}{y}}\right) + \frac{\sqrt{y}}{x} 
    \left(\frac{\frac{C}{\sqrt{y}} - \frac{1}{y}}{2\sqrt{1 - \frac{C x}{\sqrt{y}} + \frac{x}{y}}}\right).
\end{equation*}

Let $A = \sqrt{1 - \frac{C x}{\sqrt{y}} + \frac{x}{y}}$.
To look for stationary points and will set the derivative of $f$ to zero:
\[
 \frac{d}{d x}f(x;y) = 0 \iff
 -\frac{1}{x} (1-A) + \frac{\frac{C}{\sqrt{y}} - \frac{1}{y}}{2A} = 0 \iff
 x \left(\frac{C}{\sqrt{y}} - \frac{1}{y}\right) <  2A(1-A)
\]
We expand $A^2$ to get the simpler form:
\begin{align*}
    \frac{C x}{\sqrt{y}} - \frac{x}{y} = 2A - 2\left(1 - \frac{C x}{\sqrt{y}} + \frac{x}{y}\right)
    &\iff 0 = 2A - 2 + \frac{C x}{\sqrt{y}} - \frac{x}{y} = -A^2 + 2A - 1 = -(A - 1)^2 \\
    &\iff A = 1.
\end{align*}
From the definition of $A$,
\[
    A = 1 \iff
    1 - \frac{C x}{\sqrt{y}} + \frac{x}{y} = 1 \iff
    C x \sqrt{y} = x \iff
    y = \frac{1}{C^2}.
\]
As $C \geq 1$, in the parameter regime $y > 1$, $f$ has no stationary points.

It remains to check the boundary point $x = 1$ and the function $f$ in the limit as $x \to 0$.
For $x=1$, the claim is reduced to this simple inequality: 
\begin{align*}
    \sqrt{y}\left(1 - \sqrt{1 - \frac{C}{\sqrt{y}} + \frac{1}{y}}\right) \geq \frac{C}{2} - \frac{1}{2\sqrt{y}} 
    &\iff 2y\left(1 - \sqrt{1 - \frac{C}{\sqrt{y}} + \frac{1}{y}}\right) \geq C\sqrt{y} - 1 \\
    &\iff (2y - C\sqrt{y} + 1)^2 \geq  4y^2 \left(1 - \frac{C}{\sqrt{y}} + \frac{1}{y}\right) \\
    &\iff 4y^2 + C^2y + 1 - 4Cy\sqrt{y} + 4y - 2C\sqrt{y} \geq 4y^2 - 4Cy\sqrt{y} + 4y \\
    &\iff 1 + C^2y -2C\sqrt{y} \geq 0 \\
    &\iff (C\sqrt{y}-1)^2 \geq 0
\end{align*}
which holds for any value of $y$.
In the rest of the proof, we focus on the limit as $x \to 0$.

Via the Taylor expansion of \cref{prop:taylor-sqrt},
\begin{align*}
    \lim_{x \to 0} f(x;y) &= \lim_{x \to 0} \frac{\sqrt{y}}{x} \left(1 - \sum_{n=0}^\infty \frac{\prod_{k=1}^n \left(\frac{3}{2} - k\right)}{n!} \left(-\frac{Cx}{\sqrt{y}} + \frac{x}{y}\right)^n \right)\\
    &= \lim_{x \to 0} \frac{\sqrt{y}}{x} \sum_{n=1}^\infty - \frac{\prod_{k=1}^n \left(\frac{3}{2} - k\right)}{n!} \left(-\frac{Cx}{\sqrt{y}} + \frac{x}{y}\right)^n.
\end{align*}
Note that the coefficients in the summation are upper bounded in magnitude by $1$ as the sequence of terms in the descending factorial in the numerator is dominated by the sequence in the factorial in the denominator.
We also note that the absolute value of the coefficient of the first term is a constant, i.e. $\frac{1}{2}$. 
So, in the limit, the summation is dominated by the lowest order terms with respect to $x$ which correspond to $n=1$. In this case, the coefficient is $-\frac{1}{2}$ and the limit evaluates to
\begin{equation*}
    \lim_{x \to 0} f(x;y) = \frac{\sqrt{y}}{{x}} \left(\frac{Cx}{2\sqrt{y}} - \frac{x}{2y}\right) = \frac{C}{2} - \frac{1}{2\sqrt{y}}.
\end{equation*}
\end{proof}

\section{Utility Analysis}\label{sec:utility}
\subsection{Stochastic Dominance Proof}
\begin{theorem}\label{thm:maxalgo-perf}
Let $\beta \geq 0$ be the parameter controlling the adaptive threshold excess.
Let $U$ be the set of items output when using \basicalgo{} as the weighting algorithm and let $U^*$ be the set of items output when using unbiased \ouralgo{} as the weighting algorithm.
Then, for items $i \in \U$,
\begin{itemize}
    \item If $\Pr\p{i \in U} < \Phi(\beta)$, then $\Pr\p{i \in U^*} \geq \Pr\p{i \in U}$.
    \item Otherwise, $\Pr\p{i \in U^*} \geq \Phi(\beta)$.
\end{itemize}
\end{theorem}

\begin{proof}[Proof of \cref{thm:maxalgo-perf}]
Let $w$ and $w^*$ be the weights produced by \basicalgo{} and \ouralgo{}, respectively.
Let $I_{adapt} = \{u \in [n] : |S_u| \leq \dmax\}$ be the items which participate in adaptive rerouting.
For any item $i \in \U$, we will consider it's initial weight under the adaptive algorithm:
\[
    w^*_{init}(i) = \sum_{u \in I_{adapt}: i \in S_u} \frac{1}{|S_u|}.
\]
We will proceed by cases.
\paragraph{Case 1:} $w^*_{init}(i) \leq \tau$. In this case, $w^*(i) \geq w(i)$. In the adaptive algorithm, no weight is truncated from the initial weights, and so each user contributes to the final weight of an item $\frac{1}{\sqrt{|S_u|}}$ plus rerouted weight from other items. As the weight on item $i$ only increases for the adaptive algorithm compared to the basic algorithm, the probability of outputting $i$ also can only increase.

\paragraph{Case 2:} $w^*_{init}(i) > \tau$. In this case in the adaptive algorithm, the initial weight is truncated to $\tau$, excess weight is rerouted, and a final addition of $\frac{1}{\sqrt{|S_u|}} - \frac{1}{|S_u|}$ is added. As $|S_u| \geq 1$, the final weight $w^*(i) \geq \tau$.
Then, $i \in U^*$ if the added Gaussian noise does not drops the weight below the threshold $\rho$, i.e., if the noise is greater than or equal to 
\[
    \rho - \tau = \rho - (\rho + \beta \sigma) = -\beta \sigma.
\]
As the noise has zero mean and standard deviation $\sigma$, this probability is exactly $1 - \Phi(-\beta) = \Phi(\beta)$.
\end{proof}

\subsection{Example showing a gap between \ouralgo{} and parallel baselines}
While \cref{thm:maxalgo-perf} bounds the worst-case behavior of our algorithm compared to the basic algorithm, as \ouralgo{} increases the weight of items below $\tau$ compared to the basic algorithm, it will often have a larger output. We show a simple, explicit example where our algorithm will substantially increase the output probability of all but one item.

Here, we demonstrate an explicit setting where \ouralgo{} outperforms the baselines.
There are $n$ users each with degree $3$ as well as a single heavy item $i^*$ and $m$ light items. Each user's set is comprised of $i^*$ as well as two random light items. Under the basic algorithm, each user will contribute $\nicefrac{1}{\sqrt{3}}$ to each of their items. Therefore, the weights under \basicalgo{} are
\begin{equation*}
    w(i) = 
    \begin{cases}
        \frac{n}{\sqrt{3}} & \text{if } i = i^* \\
        \frac{2n}{\sqrt{3}m} < \frac{1.16n}{m} & o.w.
    \end{cases}.
\end{equation*}

On the other hand, assuming $n >> \tau$, \ouralgo{} will reroute almost all of the initial weight on the heavy item back to the users, so each user will have excess weight approximately $\nicefrac{1}{3}$. For $\dmax=3$, we get discount factor $\alpha > 0.5$. So, each user will send approximately $\nicefrac{1}{18}$ weight to each of its items. The weights under \ouralgo{} are
\begin{equation*}
    w(i) = 
    \begin{cases}
        \tau + \frac{n}{18} & \text{if } i = i^* \\
        \frac{2n}{m}\left(\frac{1}{\sqrt{3}} + \frac{1}{18}\right) < \frac{1.27n}{m} & o.w.
    \end{cases}.
\end{equation*}

In this setting, our algorithm will assign close to $10\%$ more weight to the light items  (resulting in substantially higher probability of output) compared to the basic algorithm.
If $\nicefrac{n}{m}$ is close to the true threshold $\rho$, this gap will have a large effect on the final output size. We empirically validate this for $n=\num{15000}, m=1000, \eps=1, \delta=10^{-5}$. Our algorithm returns $610$ items on average. The basic algorithm returns $519$ items while DP-SIPS with a privacy split of $5\%, 15\%, 80\%$ or $10\%, 90\%$ returns $332$ or $514$ items, respectively. In all cases, as expected, our algorithm has significantly higher average output size.

\section{Experiments}\label{sec:experiments}
\begin{table}[ht]
\footnotesize
\begin{centering}
\begin{tabular}{| l | r | r | r | }
\hline
Dataset     & Users  & Items  & Entries  \\
\hline
Higgs       & \num{2.8e5}   & \num{5.9e4}     & \num{4.6e5} \\
IMDb        & \num{5.0e4}    & \num{2.0e5}    & \num{7.6e6} \\
Reddit      & \num{2.2e5}   & \num{1.5e5}    & \num{7.9e6} \\
Finance     & \num{1.4e6}  & \num{2.7e5}    & \num{1.7e7} \\
Wiki        & \num{2.5e5}   & \num{6.3e5}    & \num{1.8e7} \\
Twitter     & \num{7.0e5}   & \num{1.3e6}   & \num{2.7e7} \\
Amazon      & \num{4.0e6}  & \num{2.5e6}   & \num{2.4e8} \\
Clueweb     & \num{9.6e8} & \num{9.4e8} & \num{4.3e10} \\
Common Crawl & \num{2.9e9} & \num{1.8e9} & \num{7.8e11} \\
\hline
\end{tabular}
\caption{Number of distinct users, distinct items, and total entries (user, item pairs). The number of entries is the sum of the sizes of all the users' sets.}
\label{table:datasets}
\end{centering}
\end{table}

\begin{table*}[ht]
\begin{centering}
\footnotesize
\begin{tabular}{| l | l  l  l  l | l  l |}
\hline
\multirow{2}{*}{Dataset} & \multicolumn{4}{c |}{\emph{Parallel Algorithms}} & \multicolumn{2}{c |}{\emph{Sequential Algorithms}} \\
 & \ouralgo{} (ours) & \ouralgotworounds{} (ours) & \basicalgo{} & DP-SIPS & PolicyGaussian & GreedyUpdate \\
\hline
Higgs & \textbf{\num{1807}}$ \; \err{13}$ & $\num{1767} \; \err{15}$ & $\num{1791} \; \err{18}$ & $\num{1743} \; \err{8}$ & $\num{1923} \; \err{18}$ & $\underline{\num{2809}} \; \err{11}$ \\
IMDb & $\num{2516} \; \err{12}$ & \textbf{\num{3369}}$ \; \err{19}$ & $\num{2504} \; \err{7}$ & $\num{3076} \; \err{16}$ & $\underline{\num{3578}} \; \err{19}$ & $\num{1363} \; \err{11}$ \\
Reddit & $\num{4162} \; \err{19}$ & \textbf{\num{6215}}$ \; \err{18}$ & $\num{4062} \; \err{21}$ & $\num{5784} \; \err{30}$ & $\underline{\num{7170}} \; \err{39}$ & $\num{6340} \; \err{16}$ \\
Finance & $\num{12759} \; \err{16}$ & \textbf{\num{17785}}$ \; \err{28}$ & $\num{12412} \; \err{50}$ & $\num{16926} \; \err{18}$ & $\num{20100} \; \err{49}$ & $\underline{\num{23556}} \; \err{27}$ \\
Wiki & $\num{7812} \; \err{12}$ & \textbf{\num{10554}}$ \; \err{41}$ & $\num{7753} \; \err{36}$ & $\num{9795} \; \err{21}$ & $\underline{\num{11455}} \; \err{21}$ & $\num{4739} \; \err{14}$ \\
Twitter & $\num{9074} \; \err{23}$ & \textbf{\num{14064}}$ \; \err{13}$ & $\num{8859} \; \err{22}$ & $\num{13499} \; \err{50}$ & $\num{15907} \; \err{30}$ & $\underline{\num{15985}} \; \err{29}$ \\
Amazon & $\num{35797} \; \err{63}$ & \textbf{\num{67086}}$ \; \err{59}$ & $\num{35315} \; \err{69}$ & $\num{66126} \; \err{57}$ & $\num{77846} \; \err{127}$ & $\underline{\num{86841}} \; \err{95}$ \\
Clueweb & $\num{34692178}$ & $\num{34533524}$ & $\num{34603077}$ & \textbf{\num{34889208}} & -- & -- \\
Common Crawl & $\num{15815452}$ & \textbf{\num{29373829}} & $\num{15734148}$ & $\num{28328613}$ & -- & -- \\
\hline
\end{tabular}
\caption{Comparison of output size of DP partition selection algorithms with $\eps = 1$, $\delta=10^{-5}$, and $\Delta_0 = 100$. A standard hyperparameter setting is fixed for each algorithm, other than DP-SIPS, where the best result is taken from privacy splits $[0.1, 0.9]$ and $[0.05, 0.15, 0.8]$. For smaller datasets, sequential algorithms are also reported as oracles and results are averaged over $5$ trials with one standard deviation reported parenthetically. For each dataset, the best parallel result is bolded and the best sequential result is underlined.}
\label{table:results}
\end{centering}
\end{table*}

\begin{figure}[ht]
    \centering

    \begin{subfigure}{0.5\linewidth}
        \includegraphics[width=0.9\linewidth]{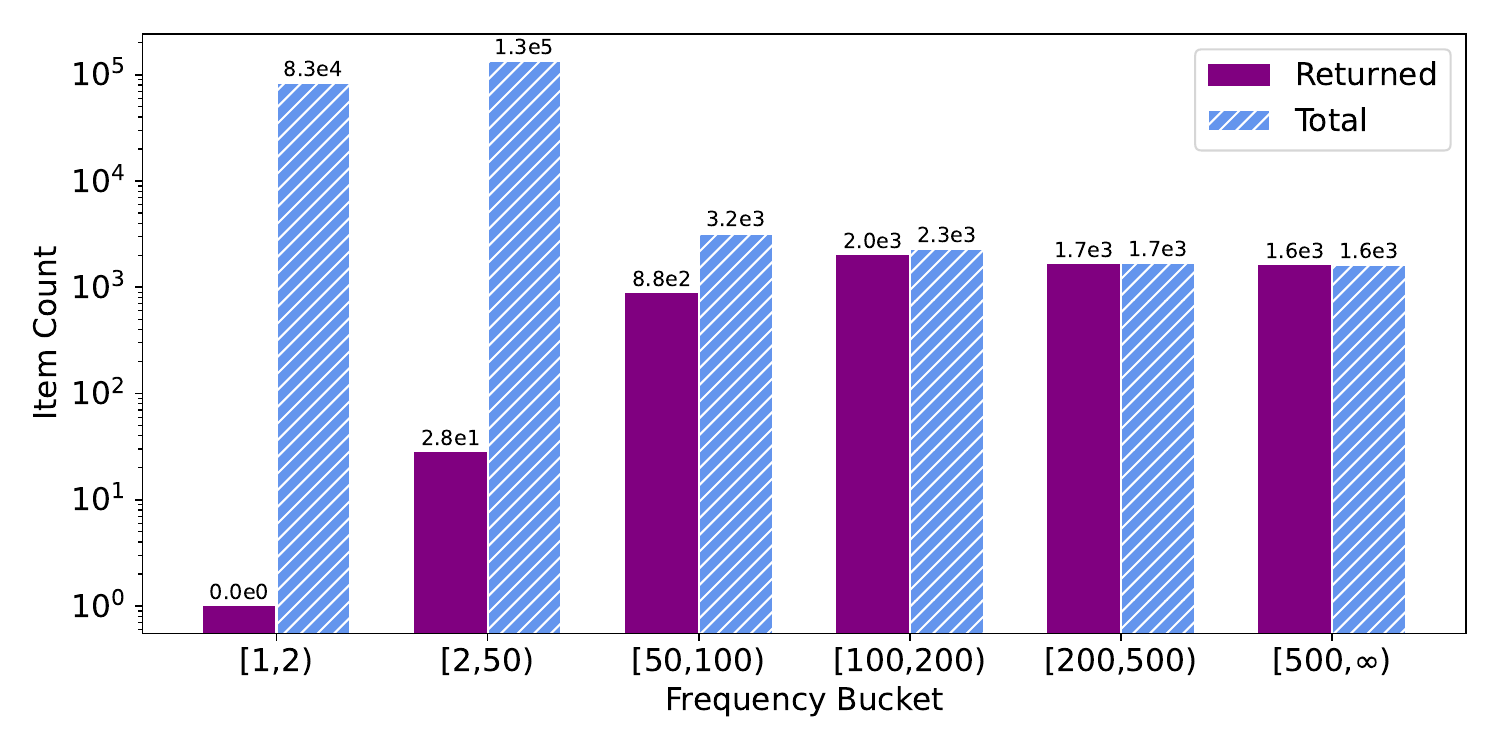}
        \caption{Reddit Item Coverage}
        \label{fig:reddit-coverage}
    \end{subfigure}

    \begin{subfigure}{0.5\linewidth}
        \includegraphics[width=0.9\linewidth]{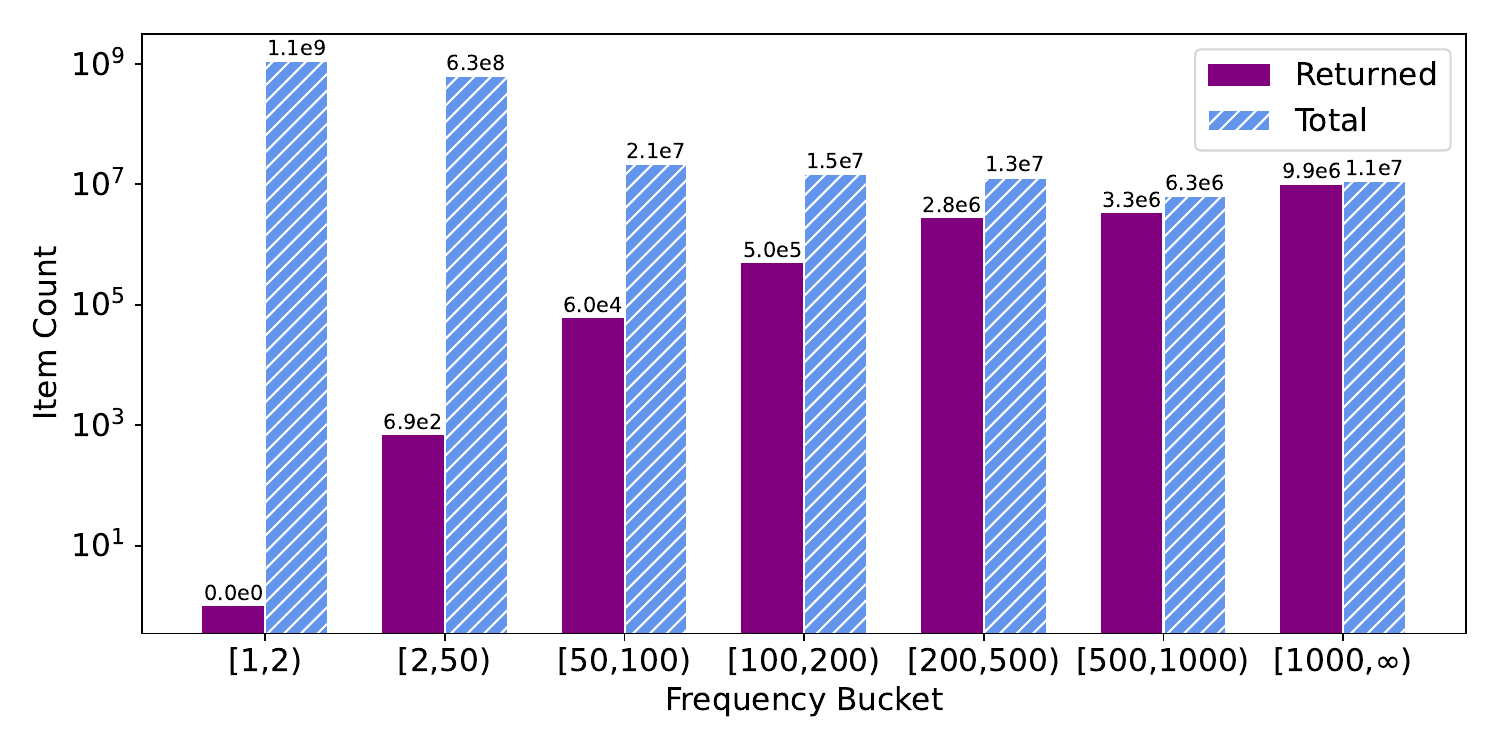}
        \caption{Common Crawl Item Coverage}
        \label{fig:commoncrawl-coverage}
    \end{subfigure}

    \caption{Comparison by item frequency of the output size of \ouralgotworounds{} to the total items on the Reddit and Common Crawl datasets. Parameters $\eps=1$ and $\Delta_0=100$ are fixed with $\delta = 10^{-5}$ for Reddit and $\delta = 10^{-11}$ for Common Crawl.}
    \label{fig:coverage}
\end{figure}

\begin{figure}[ht]
    \centering
    \includegraphics[width=0.5\textwidth]{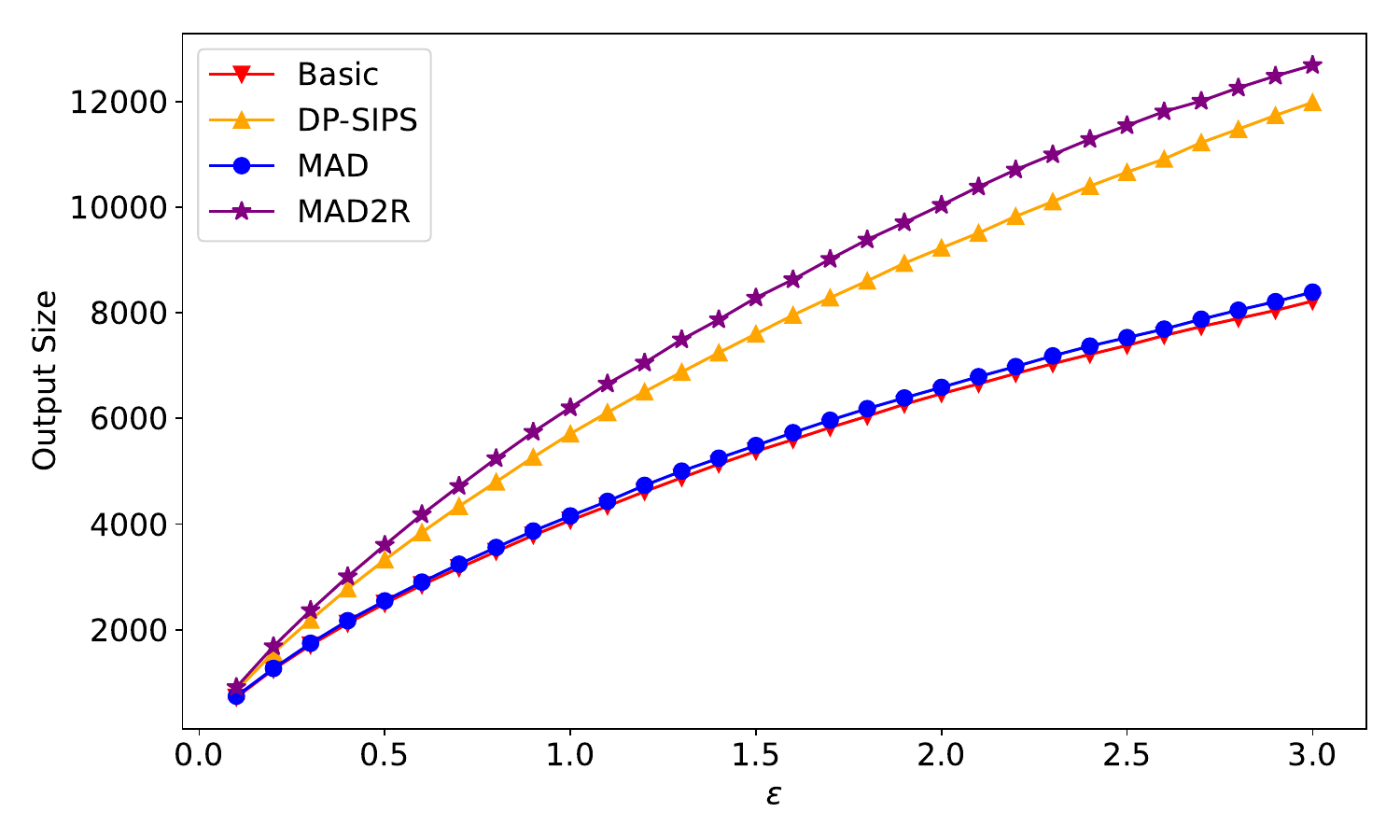}
    \caption{Comparison of output size across parallel algorithms while varying privacy parameter $\eps$ on the Reddit dataset. Other parameters are fixed as described in \cref{sec:experiment-parameters} with a fixed privacy split of $[0.1,0.9]$ for DP-SIPS and \ouralgotworounds{}. The relative performance of algorithms does not change with this parameter. Increasing $\eps$ significantly improves performance at the cost of privacy by lowering the required noise and threshold.}
    \label{fig:eps}
\end{figure}

\begin{figure}[ht]
    \centering
    \includegraphics[width=0.5\textwidth]{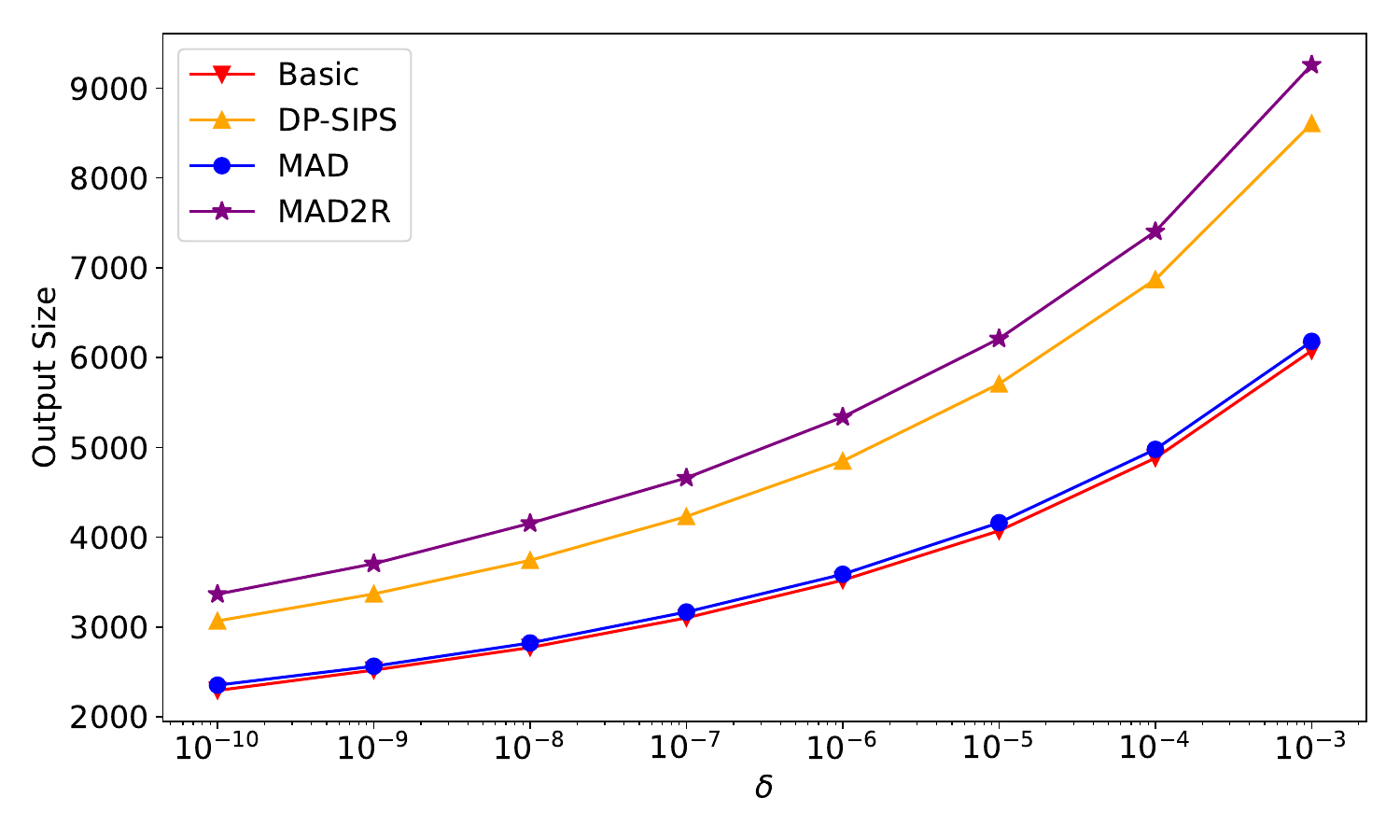}
    \caption{Comparison of output size across parallel algorithms while varying privacy parameter $\delta$ on a log-scale on the Reddit dataset. Other parameters are fixed as described in \cref{sec:experiment-parameters} with a fixed privacy split of $[0.1,0.9]$ for DP-SIPS and \ouralgotworounds{}. The relative performance of algorithms does not change with this parameter. Increasing $\delta$ significantly improves performance at the cost of privacy by lowering the required noise and threshold.}
    \label{fig:delta}
\end{figure}

\begin{figure}[ht]
    \centering
    \includegraphics[width=0.5\textwidth]{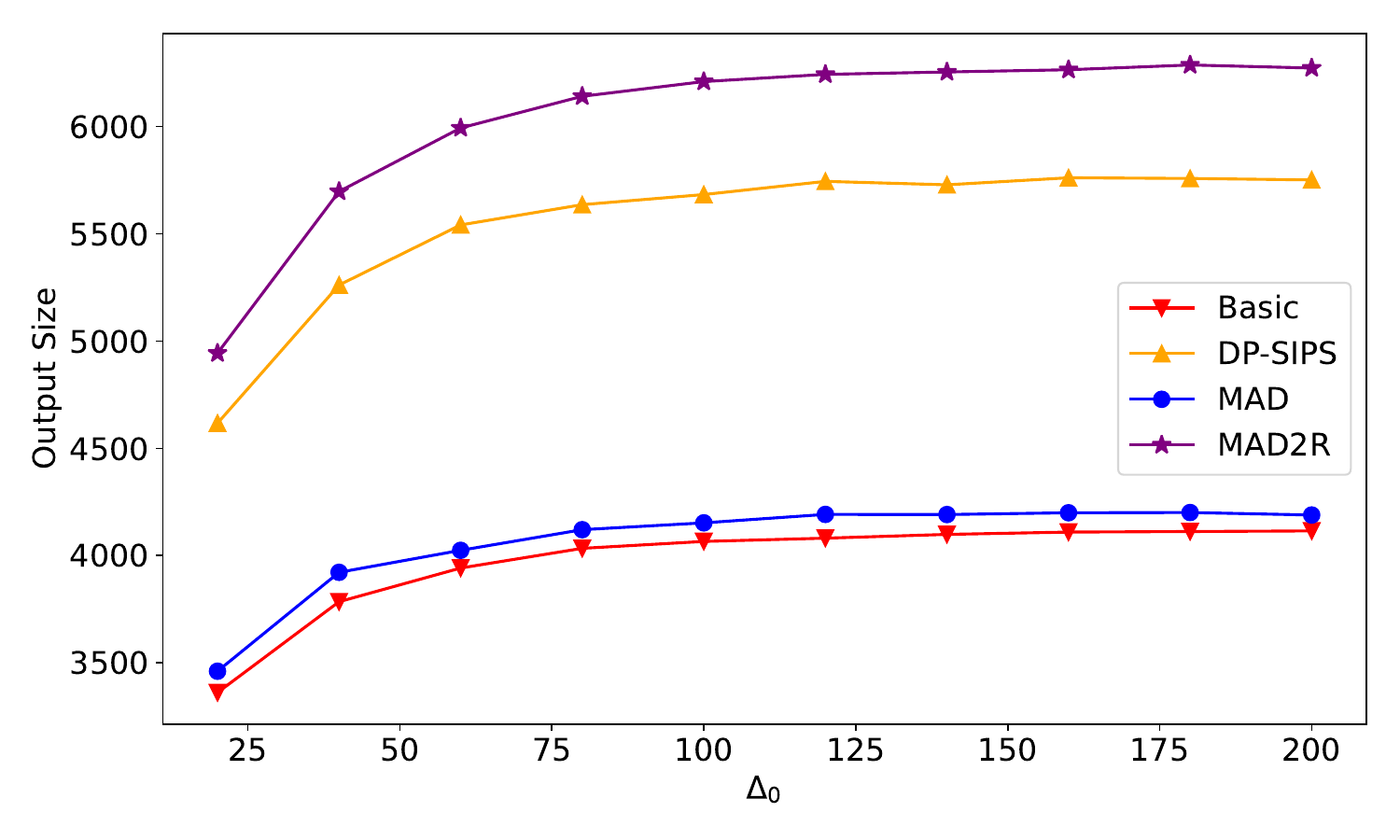}
    \caption{Comparison of output size across parallel algorithms while varying maximum set size parameter $\delta_0$ on the Reddit dataset. Other parameters are fixed as described in \cref{sec:experiment-parameters} with a fixed privacy split of $[0.1,0.9]$ for DP-SIPS and \ouralgotworounds{}. The relative performance of algorithms does not change with this parameter, and good results are achieved as long as it is not too small.}
    \label{fig:cap}
\end{figure}

\begin{figure}[ht]
    \centering
    \includegraphics[width=0.5\textwidth]{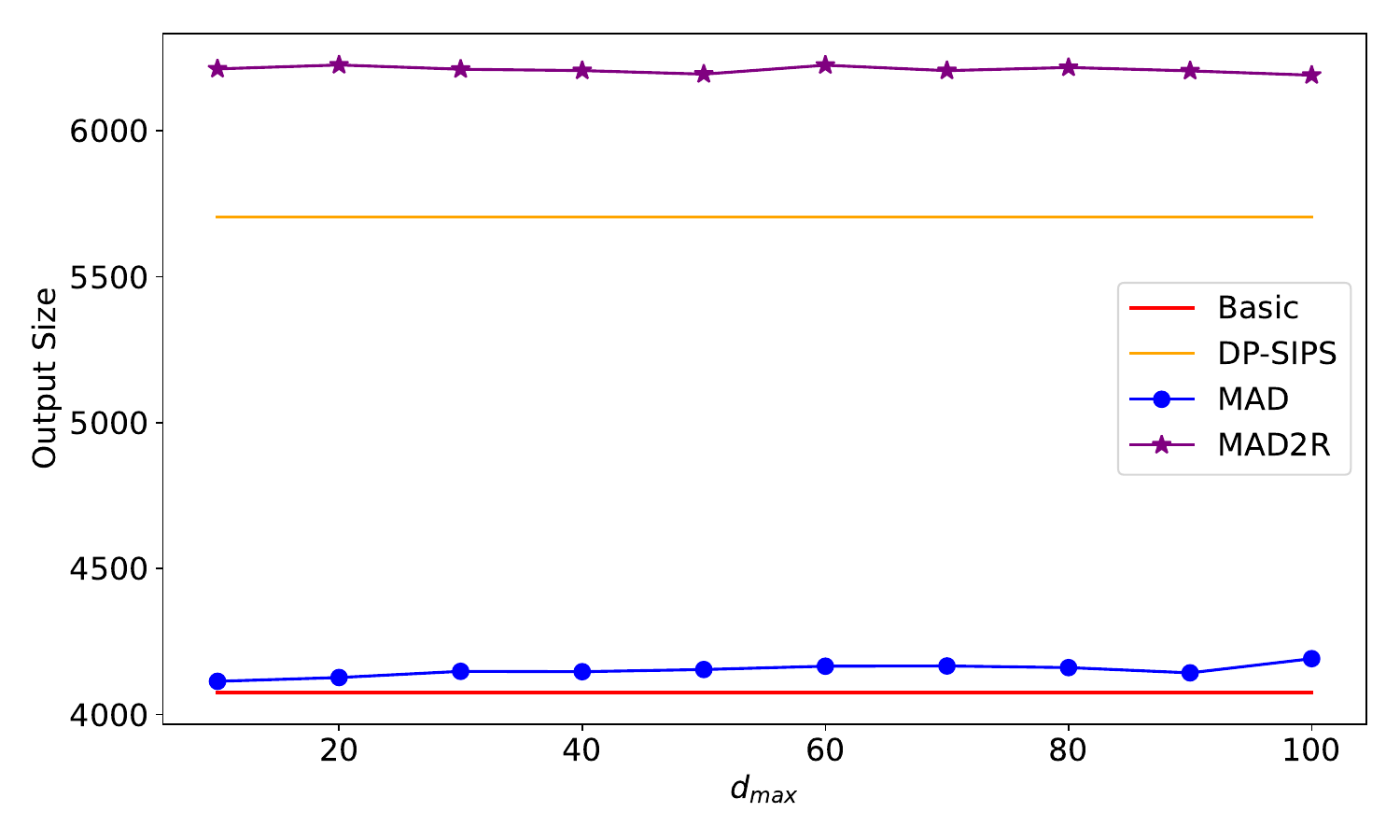}
    \caption{Comparison of output size across parallel algorithms while varying the parameter $\dmax$ of our algorithms on the Reddit dataset. As this parameter is only used by \ouralgo{} and \ouralgotworounds{}, the performance of the baselines is fixed. Other parameters are fixed as described in \cref{sec:experiment-parameters} with a fixed privacy split of $[0.1,0.9]$ for DP-SIPS and \ouralgotworounds{}. The performance of our algorithm is relatively insensitive to this parameter.}
    \label{fig:dmax}
\end{figure}
    
We now compare the empirical performance of \ouralgo{} and \ouralgotworounds{} against two parallel (\basicalgo{}, DP-SIPS) and two sequential algorithms (PolicyGaussian and GreedyUpdate) for the partition selection. We observe that our algorithms output most items (at parity of privacy parameter) among the parallel algorithms for every dataset and across various parameter regimes. Moreover, parallelization allows us to analyze datasets with up to $800$ billion entries, orders of magnitude larger than sequential algorithms. In the rest of the section, we describe the datasets, algorithms, and computational setting, before presenting our empirical results.

\subsection{Datasets}
We consider 9 datasets with statistics detailed in \cref{table:datasets}.
First, we consider small-scale datasets that are suitable for fast processing by sequential algorithms in a single-core architecture. These includes, for the sake of replicability, datasets used in prior works~\cite{gopi2020dpunion, carvalho2022incorporatingitem, swanberg2023dpsips}. These datasets have up to $3$ million distinct items and $300$ million entries.
Higgs~\cite{snapnets} is a dataset of Tweets during the discovery of the Higgs.
IMDb~\cite{imdb} is a dataset of movie reviews, Reddit~\cite{gopi2020dpunion} is a dataset of posts to \texttt{r/askreddit}, Finance~\cite{finance} is dataset of financial headlines, Wiki~\cite{wiki} is a dataset of Wikipedia abstracts, Twitter~\cite{twitter} is a dataset of customer support tweets, and Amazon~\cite{amazon1, amazon2} is a dataset of product reviews. 
For each of these text-based datasets we replicate prior methodology~\cite{gopi2020dpunion, carvalho2022incorporatingitem} where items represent the tokens used in a document and each document corresponds to a user (in some datasets, actual users are tracked across documents, in which case, we use combine the users' documents into one document).

We also consider two very-large publicly-available datasets Clueweb~\cite{BRSLLP} and Common Crawl\footnote{\url{https://www.commoncrawl.org/}}. The latter has approximately $2$ billion distinct items and $800$ billion entries. This is $3$ orders of magnitude larger than the largest dataset used in prior work.
Clueweb~\cite{BRSLLP} is a dataset of web pages and their hyper-links, items corresponds to the hyperlinks on a web page and each page corresponds to a user. 
Common Crawl is a very-large text dataset of crawled web pages often used in LLM research.

\subsection{Algorithms and Parameters}
\label{sec:experiment-parameters}

We compare our results to both sequential and parallel algorithms from prior work. The sequential algorithms we compare against are PolicyGaussian~\cite{gopi2020dpunion} and GreedyUpdate~\cite{carvalho2022incorporatingitem}. Like our algorithm, both algorithms set an adaptive threshold $\tau$ greater than the true threshold $\rho$. They try to maximize weight assigned to items up to but not exceeding $\tau$. PolicyGaussian goes through each user set one by one and adds $\ell_2$ bounded weight to minimize the $\ell_2$ distance between the current weight and the all $\tau$ vector, $w(i) = \tau \; \forall i \in \U$. GreedyUpdate goes through each user set one by one and increments the weight of a single item in the set by one, choosing an item whose weight is currently below $\tau$.\footnote{Unlike all of the other algorithm, this algorithm does not do a first step of bounding users' degrees by $\Delta_0$ as it only assigns weight to a single item per user by design.}
As observed before~\cite{swanberg2023dpsips}, sequential algorithms can have arbitrary long adpativity chains (the processing of each user can depend on all prior users processed) thus allowing larger output sizes than parallel algorithms. This, however, comes at the cost of not being parallelizable (as we observe in our experiments on the larger datasets).     
The parallel baselines we compare against are \basicalgo{}~\cite{korolova2009releasing, gopi2020dpunion} and DP-SIPS~\cite{swanberg2023dpsips}. In DP-SIPS, the privacy budget is split into a distribution over rounds. In each round, the basic algorithm is run with the corresponding privacy budget. Items found in previous rounds are removed from all user's sets for the next rounds.

We make parameter choices which are consistent with prior work and generally work well across datasets.
Unless otherwise specified, we use $\eps = 1$, $\delta=10^{-5}$, and $\Delta_0 = 100$.\footnote{We report these privacy settings for consistency with prior work in the literature, but observe the results are consistent across various choices. For real production deployments on large-scale sensitive data, $\delta$ is usually smaller.}
For PolicyGaussian and GreedyUpdate, we set the $\beta=4$ to be the number of standard deviations of noise to add to the base threshold to set the adaptive threshold.
For DP-SIPS, we take the best result of running with a privacy split of $[0.1, 0.9]$ and $[0.05, 0.15, 0.8]$\footnote{As this choice can have a significant effect on performance, we choose the best-performing to give this baseline the benefit of the doubt.}.
For \ouralgo{} and \ouralgotworounds{}, we set $\dmax=50$ and $\beta=2$.
For \ouralgotworounds{}, we set the privacy split of $[0.1,0.9]$, $\bmin=0.5$, $\bmax=2$, $C_{lb}=1$, and $C_{ub}=3$.

\subsection{Computing Details}
We perform experiments in two different computational settings. First we implement a sequential, in-memory version of all algorithms (including the parallel ones) using Python.\footnote{An open-source Python implementation of our algorithm is available at \url{https://github.com/jusyc/dp_partition_selection}.}.
For PolicyGaussian and GreedyUpdate we use the Python implementations from prior work~\cite{gopi2020dpunion, carvalho2022incorporatingitem}. This allows us to fairly test the scalability of the algorithms not using parallelism. As we observe next, this approach does not scale to the two largest datasets we have (Clueweb, Common Crawl). 

Then, we implement all parallel algorithms (\ouralgo{}, \ouralgotworounds{}, \basicalgo{}, DP-SIPS) using C++ in a modern multi-machine massively parallel computation framework in our institution. This framework allows to use a fleet of shared (x86\_64) architecture machines with 2.45GHz clocks. The machines are shared by several projects and can have up to 256 cores and up to 512GB of RAM. The jobs are dynamically allocated RAM, machines and cores depending on need and availability. 
As we observe, all parallel algorithm are very scalable and run on these huge datasets within 4 hours of wall-clock time. On the other hand, both sequential algorithms cannot exploit this architecture and could not complete in 16 hours on the Clueweb dataset (we estimate they would take several days to complete on the Common Crawl dataset even assuming access to enough memory).

\subsection{Results}
\paragraph{Algorithm Comparison}
\cref{table:results} displays the output size of the DP partition selection algorithms (i.e., the number of privatized items output).
Among parallel algorithms, \ouralgotworounds{} achieves the best result on seven out of nine datasets. The two exceptions are the Higgs dataset, where \ouralgo{} performs the best, and the Clueweb dataset, where DP-SIPS performs the best. Both of these datasets have outlier statistics (see \cref{table:datasets}): the average size of a user set in the Higgs dataset is less than $2$ and the number of unique items in the Clueweb dataset is less than the number of users.
Directly comparing \ouralgo{} with \basicalgo{}, \ouralgo{} is always better, corroborating our proof of stochastic dominance.
Comparing \ouralgotworounds{} with DP-SIPS, \ouralgotworounds{} is almost always significantly better, by up to a factor of a $9.5\%$ improvement on the IMDb dataset.

On the small scale datasets where we can run sequential algorithms, as expected from prior work~\cite{swanberg2023dpsips}, one of the two sequential algorithms yield the best results across all algorithms with PolicyGaussian consistently outperforming all parallel baselines. GreedyUpdate's performance is heavily dataset dependent, sometimes performing the best and sometimes the worst out of all algorithms.  This is not a surprise as the sequential algorithms utilize much more adaptivity than even our adaptive parallel algorithm at the cost of limiting scalability. Our algorithm is still competitive, never outputting fewer than $86\%$ of the items of PolicyGaussian (and outperforming GreedyUpdate on many datasets). 
For massive datasets, where it is simply infeasible to run the sequential algorithms, however \ouralgotworounds{} has the best results of all parallel algorithms.

Figures comparing output sizes while varying $\eps$, $\delta$, $\Delta_0$, and $\dmax$ are included in \crefrange{fig:eps}{fig:dmax}. The relative performance of the algorithms is the same across many choices.

\paragraph{Absolute Utility}
To understand the absolute utility of our algorithms (as opposed to relative to other baselines), we focus on the performance of \ouralgotworounds{} on the Reddit and Common Crawl datasets.
In order to understand the performance in a real deployment rather than compare baselines across common parameter settings, we change the $\delta$ for the large scale Common Crawl dataset to $\delta=10^{-11}$.

On the Reddit dataset, \ouralgotworounds{} outputs $\num{6340}$ out of $\num{143556}$ unique items ($4.4\%$).
On the other hand, $98\%$ of users have at least one outputted item, and $45\%$ of the entries (user-item pairs) belong to an item which is output by our algorithm.
The relatively small overall fraction of items output is due in part to the fact that the Reddit dataset has $58\%$ singleton items (items only appearing in a single user’s set). Any algorithm which outputs any singleton items is not private, as it is leaking private information belonging to a single user. For any algorithm with acceptable privacy settings, outputting items with very small frequencies is also simply not possible.
In \cref{fig:reddit-coverage}, we break down the number of items total in the dataset and the number output by \ouralgotworounds{} broken down by item frequency.
Our algorithm returns almost all of the items with frequency at least $100$.

On the Common Crawl dataset, \ouralgotworounds{} outputs $\num{16551550}$ out of $\num{1803720630}$ unique items ($0.9\%$). 
On the other hand, $99.9\%$ of users have at least one outputted item and $97\%$ of entries in the dataset belong to an item in output by our algorithm. This dataset contains $61\%$ singleton items, and many low frequency items.
In \cref{fig:commoncrawl-coverage}, we break down the number of items total in the dataset and the number output by \ouralgotworounds{} broken down by item frequency.
Our algorithm returns an overwhelming fraction of items occuring in at least $200$ user sets on a dataset with billions of users overall.

\section{Conclusion}
We introduce \ouralgo{} and \ouralgotworounds{}, new parallel algorithms for private partition selection which provide state-of-the-art results, scale to massive datasets, and provably outperform baseline algorithms.
Closing the remaining gap between parallel and sequential algorithms remains an interesting direction, as well as developing new ideas to adaptive route weight to items below the privacy threshold while maintaining bounded sensitivity.
While we are able to prove \emph{ordinal} theoretical results (our algorithm is at least as good as another), it is an open challenge to develop a framework where we can prove \emph{quantitative} results, perhaps comparing the competitive ratio of a private partition selection algorithm compared to some reasonably defined optimum.

\clearpage
\section*{Acknowledgements}
Justin Chen is supported by an NSF Graduate Research Fellowship under Grant No. 17453.

\bibliographystyle{alpha}
\bibliography{bib}


\clearpage
\appendix
\section{Basic Algorithm}

\begin{algorithm}[!ht]
\caption{\basicalgo
 \newline  {\bf Input: } User sets $\sets = \{(u, S_u)\}_{u \in [n]}$
 \newline  {\bf Output: } $w: \U \to \mathbb{R}$ weighting of the items
}
\begin{algorithmic}[1]
\label{alg:basic}
\STATE Initialize weight vector $w$ with zeros
\FORALL{$u \in [n]$}
    \STATE $w(i) \pluseq 1/\sqrt{|S_u|}$ for $i \in S_u$ \COMMENT{Add basic $\ell_2$ bounded weight.}
\ENDFOR
\RETURN $w$
\end{algorithmic}
\end{algorithm}


\end{document}